\documentclass[review=false,authordraft=false]{jfp-epi}

\usepackage{colortbl}
\usepackage{array}
\usepackage{booktabs}
\usepackage{mathpartir}
\usepackage{listings}
\usepackage{needspace}

\definecolor{kwColor}{RGB}{0,114,178}
\definecolor{strColor}{RGB}{150,30,30}
\definecolor{cmtColor}{RGB}{105,105,105}
\lstdefinelanguage{Rust}{
  morekeywords={fn,pub,let,mut,while,use,impl,struct,enum,match,if,else,
    return,move,for,in,loop,ref,as,const,static,trait,type,where,unsafe,
    self,Self,true,false,mod,crate,super,dyn,break,continue,proof,spec,
    exec,requires,ensures,invariant,decreases,forall,exists,implies,by,
    verus},
  sensitive=true,
  morecomment=[l]{//},
  morestring=[b]",
  moredelim=[s][\color{sensColor}]{\#[spec(}{)]},
  moredelim=[s][\color{sensColor}]{\#[refined\_by(}{)]},
  moredelim=[l][\color{cmtColor}]{\#[trusted]},
  moredelim=[l][\color{cmtColor}]{\#[opaque]},
}
\definecolor{typeColor}{RGB}{0,114,178}
\definecolor{sensColor}{RGB}{0,149,3}
\definecolor{privColor}{RGB}{199,0,124}
\newcommand{\sens}[1]{{\color{sensColor}#1}}
\newcommand{\priv}[1]{{\color{privColor}#1}}
\everymath{\color{typeColor}}
\everydisplay{\color{typeColor}}
\newcolumntype{L}{>{\normalcolor}l}
\newcolumntype{C}{>{\normalcolor}c}
\newcolumntype{R}{>{\normalcolor}r}
\newcolumntype{P}[1]{>{\normalcolor}p{#1}}
\newcolumntype{N}{>{\color{typeColor}}r}
\newcolumntype{M}{>{\color{typeColor}}c}
\arrayrulecolor{black}

\newcommand{\sys}[1]{\textsc{#1}}
\newcommand{\forte}{\sys{Forte}}
\newcommand{\kwus}{\textunderscore\penalty500\relax}
\DeclareRobustCommand{\kw}[1]{\textup{\texttt{\let\_\kwus#1}}}
\newcommand{\cellstack}[1]{\begin{tabular}[t]{@{}L@{}}#1\end{tabular}}
\newcommand{\rulename}[1]{\textsc{#1}}

\makeatletter
\AtBeginDocument{%
  \fancypagestyle{firstpagestyle}{%
    \fancyhf{}%
    \renewcommand{\headrulewidth}{\z@}%
    \renewcommand{\footrulewidth}{\z@}%
    \fancyhead[LO]{\ACM@linecountL}}}
\makeatother

\begin{document}

\title{Forte: A sensitivity type system for imperative Rust}

\author{Chik\'e Abuah}

\begin{abstract}
We introduce \forte{}, a sensitivity type system for Rust whose soundness
rests on ownership. The graded sensitivity type systems, from Fuzz's
linear grading to Solo's environment indices, are pure calculi: a claim about a value holds for the value's whole lifetime
because nothing can mutate it. The imperative sensitivity analyses admit
assignment to first-order variables and no references, so no question of
aliasing arises in them. The programs that compute differentially private
statistics in deployment are Rust, and they mutate through borrows.
\forte{} closes this gap. Its central rules strongly update a sensitivity
environment through an exclusive borrow, at a primitive call and across a
checked function boundary; its soundness theorem is metric preservation
over an operational semantics with a store, in which the exclusivity of
\kw{\&mut} alone licenses framing across a mutating call, and two aliased
borrows suffice to refute the theorem without it.
Verus mechanizes the theorem, the function rule, and the refutation.
Flux checks \forte{} as an ordinary library, with no fork of the compiler;
a machine-checked theorem backs every deterministic primitive signature,
and a correspondence theorem transports metric preservation to the
programs the checker accepts. We evaluate \forte{} on
mechanism kernels from OpenDP with genuine in-place mutation, matching the
library's trusted stability maps with checked constants, covering the
constructors that have no proof document, rejecting off-by-one diameters,
tightened bounds, miscalibrated releases, and overspent budgets, and
deriving one trusted constant as an inferred loop invariant.
\end{abstract}


\maketitle

\section{Introduction}
\label{sec:intro}

Differential privacy~\cite{dmns06,dworkroth} is the gold standard for
privacy-preserving data analysis, and every differentially private
mechanism calibrates its noise to a \emph{sensitivity}: a bound on how far
a query's output can move when one individual's data changes. The privacy
guarantee touches the program at the sensitivity. A noise scale calibrated to an underestimated sensitivity yields
a mechanism that claims more protection than it provides, and the error
produces no symptom: the mechanism runs and the output looks plausible
while the guarantee no longer holds.

The record of sensitivity claims in deployed libraries motivates static
enforcement. Casacuberta et al.~\cite{csvw22} document widespread
sensitivity underestimation across differential-privacy implementations,
much of it traceable to the subtleties a type system should police: metric
conventions, neighboring definitions, and constants conflated between
them. The type-system literature proves its core calculi sound and
asserts its primitive libraries.

However, the type systems that could police these claims cannot host the
deployed code. Sensitivity type systems track sensitivity as a grade on
functional types, whether through graded linear types
(Fuzz~\cite{fuzz}, DFuzz~\cite{dfuzz}, Duet~\cite{duet}), contextual
linear types (Jazz~\cite{jazz}), or environment indices
(Solo~\cite{solo}), and all are pure calculi: values flow by substitution,
and a claim attached to a value never needs revisiting, because nothing
can mutate what it describes. The imperative sensitivity analyses,
Fuzzi~\cite{fuzzi}, the LightDP line~\cite{lightdp,shadowdp,checkdp}, and
the continuity analyses of Chaudhuri et al.~\cite{chaudhuri10,chaudhuri11},
admit assignment to first-order variables, which retypes the variable, and
no references, so two names never denote one location and no question of
aliasing arises. The code that computes differentially private statistics
in deployment is imperative Rust, and it mutates through references.
OpenDP~\cite{opendp}, the flagship open-source library, attaches to each
transformation a \emph{stability map}, a trusted runtime closure from
input distance to output distance. 109 hand-written proof documents
accompany the codebase; they reach 14 of the 68 trusted stability
closures in its transformation constructors, and every variant of the
bounded sum is among the closures without one (\S\ref{sec:casestudies}).
Mechanism kernels accumulate totals through mutable borrows, build
histograms by in-place increment, and grow tree layers in place. No
existing sensitivity type system can express these programs, and the gap
is structural: a borrow passed to a function is a second name for a
location the caller still owns, and a claim attached to the location must
survive the call.

We introduce \forte{}, a sensitivity type system for Rust whose soundness
rests on ownership. Our key insight is that the exclusivity of a Rust
mutable borrow alone licenses \emph{framing} across a mutating call, so
the type system can strongly update a sensitivity environment in
place where prior systems required purity or linearity. The insight has a
converse. Sensitivity environments make the update arithmetic itself
alias-tolerant, since index addition prices a double read of one
location; aliasing instead breaks the frame, the treatment of distinct
claims as claims about distinct locations, and two aliased exclusive
borrows suffice to derive a false claim (\S\ref{sec:metatheory}).
The soundness proof therefore consumes ownership at one point and asks
nothing substructural of data: weakening, contraction, and exchange are
all admissible in the sensitivity layer.

\kw{forte-rs} realizes \forte{} as an ordinary library, and
Flux~\cite{flux} checks it, with no compiler fork. Sensitivity
environments are refinement indices on opaque structs; strong updates are
Flux's existing \kw{\&mut} ensures clauses; Flux infers loop invariants. Flux's own
soundness theorem is unary, a statement about one run, while a
sensitivity claim relates two runs; metric preservation is the theorem
Flux cannot supply, and a correspondence theorem
(\S\ref{sec:implementation}) transports it to every program the checker
accepts in the fragment. Verus~\cite{verus} checks the trusted primitive
signatures: every deterministic signature obligation is a Verus theorem,
and Verus mechanizes the soundness core of the calculus, including the
function rule and the refutation of its aliased form. We evaluate \forte{} on verified re-implementations of OpenDP
mechanism kernels chosen for their in-place mutation, checking certified
constants against the library's own stability maps, and in one case
certifying a constant OpenDP asserts as trusted arithmetic.

\paragraph{Threat model.}
We assume an ``honest but fallible'' mechanism author: the programmer
intends to produce a differentially private program, but may
unintentionally introduce bugs, in particular wrong sensitivity constants
and wrong metric conventions. We do not defend against a malicious
author, and the probabilistic release step (noise addition, budget
accounting beyond affine token discipline) is outside the checked fragment
(\S\ref{sec:implementation}).

\paragraph{Arithmetic model.}
We idealize values: integers are unbounded and reals exact, so overflow
and floating-point rounding fall outside the guarantee. The
constants \forte{} certifies are the constants OpenDP's proof documents
argue about in the same idealized model; OpenDP's own treatment of
overflow (its checked and saturating sum variants) and of floating-point
error (additive relaxation terms in stability maps) composes with the
checked constant as a trusted additive layer, and the floating-point half
of the failures Casacuberta et al.\ report~\cite{csvw22} is out of scope
here.

\paragraph{Contributions.}
In summary, this paper makes the following contributions:
\begin{itemize}
\item $\lambda$-\forte{}, the first sensitivity type system for a
  language with mutation through exclusive references, with strong
  updates of sensitivity environments at primitive calls and across
  checked function boundaries (\S\ref{sec:calculus});
\item metric preservation restated and proved over a store-based
  operational semantics, identifying borrow exclusivity as the single
  ownership fact soundness consumes, with a counterexample showing its
  necessity; Verus mechanizes the theorem, the function rule, and the
  counterexample (\S\ref{sec:metatheory});
\item \kw{forte-rs}, a Flux-checked realization with a correspondence
  theorem relating accepted programs to $\lambda$-\forte{} derivations,
  and a trusted base that Verus itself checks: a development of 43
  verified theorems and lemmas covering every deterministic signature
  obligation, including the diameter and metric-conversion rules
  Casacuberta et al.\ show to be error-prone (\S\ref{sec:implementation});
\item a case study on OpenDP mechanism kernels with in-place mutation
  that turns the library's trusted stability constants into checked
  obligations, including the 54 constructors that have no proof document,
  rejects the errors Casacuberta et al.\ found in deployed libraries
  (off-by-one diameters, tightened bounds, and the Lipschitz reading of a
  clamp) together with under-calibrated releases and overspent budgets,
  and derives one trusted constant as an inferred loop invariant
  (\S\ref{sec:casestudies}).
\end{itemize}

\paragraph{Roadmap.}
The rest of the paper proceeds as follows. \S\ref{sec:overview}
introduces \forte{} by example, from OpenDP's trusted stability map to the
checked, mutating counterpart. \S\ref{sec:calculus} presents the
$\lambda$-\forte{} calculus and \S\ref{sec:metatheory} its metatheory.
\S\ref{sec:implementation} describes the implementation, the
correspondence theorem, and the trusted base. \S\ref{sec:casestudies}
reports the case studies, and \S\ref{sec:byexample} verifies OpenDP's
own primitive bodies and stability maps by example. Finally, we survey
related work (\S\ref{sec:related}) and conclude (\S\ref{sec:conclusion}).

\section{Overview of Forte}
\label{sec:overview}

Consider the most common transformation in deployed differential privacy:
a program clamps a dataset of integers to a public interval $[lo, hi]$,
sums it, and releases the sum with noise calibrated to the sum's
sensitivity. In
OpenDP, the sensitivity claim for this pipeline is code. The constructor
\kw{make\_sized\_bounded\_int\_checked\_sum} attaches to the transformation
a \emph{stability map}, a closure from input distance to output distance,
reproduced verbatim from the library:\footnote{OpenDP commit
\kw{4198f76}, September 2026, throughout.}

\begin{lstlisting}
StabilityMap::new_fallible(
    // If d_in is odd, we still only consider databases with
    // (d_in - 1) / 2 substitutions, so floor division is acceptable
    move |d_in: &IntDistance|
        T::inf_cast(d_in / 2).and_then(|d_in| d_in.inf_mul(&range)),
)
\end{lstlisting}

\noindent
OpenDP trusts the closure, and its constant is correct. The repository
holds 109 hand-written proof documents for such claims, but no document
exists for this constructor: its justification is the comment above and a
citation, and the documents overall reach 14 of the 68 trusted stability
closures in the transformation constructors (\S\ref{sec:casestudies}). The
correctness of a trusted constant is an empirical fact about a review
process, and the record of such constants across the ecosystem is poor:
Casacuberta et al.\ document widespread sensitivity underestimation across
deployed libraries~\cite{csvw22}.

In \forte{}, the same pipeline is a checked composition. The dataset is
\emph{sized}, meaning its length is public and fixed, so a neighboring
dataset substitutes rows without adding or removing any. The dataset
has a \emph{sensitivity environment} in its type, one coefficient per
declared source; clamping establishes a bounds invariant; and the sum rule
for bounded, sized data multiplies the environment by the interval's
\emph{diameter}:

\begin{lstlisting}
#[spec(fn(data: &SVec[@s], lo: i64, hi: i64{lo <= hi})
       -> SInt[(hi-lo)*s.sa, (hi-lo)*s.sb, (hi-lo)*s.sc])]
pub fn sized_bounded_sum(data: &SVec, lo: i64, hi: i64) -> SInt {
    let clamped = SVec::clamp(data, lo, hi);
    SVecBnd::sum(&clamped)
}
\end{lstlisting}

\noindent
An SMT solver discharges the constant $\mathit{hi} - \mathit{lo}$ as a
subtyping obligation at every use, and a machine-checked theorem backs
the leaf fact behind the sum rule (\S\ref{sec:implementation}). The
two conventions agree: OpenDP's $d_{\mathit{in}}$ counts symmetric
distance, which is 2 per substitution on sized data, so its
$\lfloor d_{\mathit{in}}/2 \rfloor \cdot (U - L)$ is the
per-substituted-row $\mathit{hi} - \mathit{lo}$ above. The checker rejects a claim of
$1$-sensitivity for this pipeline, the Lipschitz reading of a clamp that
Casacuberta et al.\ identify as a recurring confusion. It rejects the sum
without the clamp as well: the unbounded vector type has no sum rule at
all, which makes the bug class of~\cite{csvw22} unrepresentable, beyond
detecting it.

\paragraph{Sensitivity under mutation.}
Any of the pure sensitivity calculi could have expressed the pipeline
above. None could have hosted it in \emph{deployment}, because of the
surrounding program. Mechanism code in Rust accumulates: it adds partial
results into a running total through a mutable borrow, builds histograms
by in-place increment, pushes tree layers onto a growing vector, and
writes the accumulation once, as a function that it calls from many
places. Consider a helper that clamps and sums one dataset and adds the
result into whichever accumulator its caller lends it:

\begin{lstlisting}
#[spec(fn(acc: &mut SInt[@av], data: &SVec[@s],
          lo: i64, hi: i64{lo <= hi})
       ensures acc: SInt[av.a + (hi-lo)*s.sa,
                         av.b + (hi-lo)*s.sb,
                         av.c + (hi-lo)*s.sc])]
pub fn add_clamped_sum(acc: &mut SInt, data: &SVec, lo: i64, hi: i64) {
    let s = sized_bounded_sum(data, lo, hi);
    SInt::add_assign(acc, s);
}
\end{lstlisting}

\noindent
The signature has an \kw{ensures} clause: the borrowed accumulator leaves
the call with an environment computed from the environment it arrived
with. The checker checks the body once against this signature, with
\kw{acc} standing for whichever local a caller borrows; inside it, the trusted
\kw{add\_assign} performs the same kind of update at a primitive. A caller
accumulates three clamped sums, one per source, into one total:

\begin{lstlisting}
#[spec(fn(x: &SVec[1,0,0], y: &SVec[0,1,0], z: &SVec[0,0,1],
          lo: i64, hi: i64{lo <= hi})
       -> SInt[hi-lo, hi-lo, hi-lo])]
pub fn sum_three_datasets_inplace(x: &SVec, y: &SVec, z: &SVec,
                                  lo: i64, hi: i64) -> SInt {
    let mut acc = SInt::constant(0);
    add_clamped_sum(&mut acc, x, lo, hi);
    add_clamped_sum(&mut acc, y, lo, hi);
    add_clamped_sum(&mut acc, z, lo, hi);
    acc
}
\end{lstlisting}

\noindent
Each call performs a \emph{strong update}: the accumulator's sensitivity
environment after the call differs from its environment before, moving
from $\sens{(0,0,0)}$ through $\sens{(\mathit{hi}-\mathit{lo},0,0)}$ to
the final index.\footnote{Green marks sensitivity material and magenta privacy
material, following the 24-Color Palette from
\url{http://mkweb.bcgsc.ca/colorblind}; the two remain distinguishable
under deuteranopia.} In a pure calculus this pattern is inexpressible. In
an imperative calculus without references it is expressible only by
inlining the helper at every call, because a function that mutates an
argument needs a name for the caller's location. The soundness argument
for the call needs one fact about the program: while \kw{\&mut acc} is
live, no other view of \kw{acc} exists. Rust's borrow checker supplies
that fact.

\paragraph{The role of exclusivity.}
The update arithmetic is alias-tolerant: a call that reads one location
twice computes $x + x$, and the environment $2s$ the rules assign is a
correct bound, because environment-indexed sensitivity charges
contraction by addition. Aliasing instead breaks the \emph{frame}.
Consider a function taking two exclusive borrows whose signature declares
the second borrow unchanged. Applied to two borrows of one local, the type
system retains the second parameter's stale environment for a location
whose contents the first parameter updated, and a later read at the stale
environment is unsound. In \kw{forte-rs} the aliased call
\kw{double\_into(\&mut acc, \&mut acc)} is a borrow-check error before the
refinement checker runs. The exclusivity of \kw{\&mut} licenses treating
each claim as a claim about a distinct location, and
\S\ref{sec:metatheory} makes this premise the content of the soundness
theorem.

\paragraph{Loops and exact bounds.}
Accumulation in practice happens under a loop, and the bound a mechanism
author needs couples the environment to the trip count. The following
function accumulates one clamped daily release per iteration, through the
helper above; Flux checks it and infers its loop invariant:

\begin{lstlisting}
#[spec(fn(n: usize, lo: i64, hi: i64{lo <= hi && hi - lo <= 100})
       -> SInt{v: v.a <= 100 * n && v.b == 0 && v.c == 0})]
pub fn sum_n_daily_releases(n: usize, lo: i64, hi: i64) -> SInt {
    let mut acc = SInt::constant(0);
    let mut i = 0;
    while i < n {
        let day = SVec::from_source_a(load_day(i));
        add_clamped_sum(&mut acc, &day, lo, hi);
        i += 1;
    }
    acc
}
\end{lstlisting}

\noindent
Flux discovers the invariant coupling the environment to the counter.
Exact bounds of this shape are available because the environment here is
an integer-sorted index sharing a sort with \kw{n};
\S\ref{sec:implementation} describes this design constraint and the one
loop in our case studies whose invariant needs a hint.

\paragraph{Element writes at public indices.}
A borrow of one element is a strong update of one coordinate of a vector's
metric. A weekday report accumulates $n$ daily totals into 7 slots in
place:

\needspace{8\baselineskip}
\begin{lstlisting}
    while i < n {
        let daily = SInt::from_source_a(fetch_day_total(i));
        SVecL1::add_at(&mut report, i % 7, daily);
        i += 1;
    }
\end{lstlisting}

\noindent
The write index is public, so both runs of the program write the same
position, and the report's $L_1$ environment grows by the added value's
environment; Flux infers the invariant with no hints. A data-dependent
index would put the loop outside the checked fragment
(\S\ref{sec:casestudies}). The formal account of strong updates, function
boundaries, and loop invariants follows in \S\ref{sec:calculus}.

\section{Formal description of Forte}
\label{sec:calculus}

$\lambda$-Forte is a first-order, store-based core calculus for the checked
composition layer of \S\ref{sec:overview}: mutable locals, exclusive borrows
at call sites, functions whose parameters are exclusive borrows, loops over
public counters, and a trusted signature set $\Sigma$ standing for the
primitive library. We chose the fragment so that the ownership question is
live: the proof must establish a frame across a mutating call, and across
a checked function boundary, while remaining an induction on evaluation
(\S\ref{sec:metatheory}). The fragment excludes borrows that outlive a
call, references stored in data structures, closures over sensitive data,
and sensitive control flow; \S\ref{sec:metatheory} discusses what each
exclusion defers. Figure~\ref{fig:syntax} gives the syntax,
Figure~\ref{fig:typing} the typing rules, and Figure~\ref{fig:semantics}
the evaluation rules that write the store.

\begin{figure}
\small
\[
\begin{array}{lrcl}
\text{values} & v & ::= & r \mid () \mid \kw{true} \mid \kw{false} \mid
  \langle v, \ldots, v \rangle \\
\text{expressions} & e & ::= & v \mid x \mid e \oplus e \mid
  \kw{let}\ x = e\ \kw{in}\ e \mid x := e \mid e;\, e \\
& & \mid & \kw{if}\ e\ \kw{then}\ e\ \kw{else}\ e \mid
  \kw{while}\ e\ \kw{do}\ e \mid p(\bar{a}) \mid f(\bar{a}) \\
\text{arguments} & a & ::= & e \mid \&x \mid \kw{\&mut}\,x \\
\text{types} & \tau & ::= & \mathsf{pub}\,B \mid S[\sens{\vec{s}}] \mid
  \mathsf{Vec}_{L1}[\sens{\vec{s}}] \mid
  \mathsf{Vec}_{\mathsf{Sub}}[lo,hi,\sens{\vec{s}}] \mid
  \mathsf{Vec}_{\mathsf{Sub}}^{\circ}[\sens{\vec{s}}] \mid
  \mathsf{Vec}_{L\infty}[\sens{\vec{s}}] \\
\text{signatures} & \mathit{Sig} & ::= & \forall\vec{\alpha}.\
  (\kw{\&mut}\,\tau_1, \ldots, \kw{\&mut}\,\tau_n,
   \tau'_1, \ldots, \tau'_m) \to \tau\
  \mathsf{ensures}\ x_j : \tau^{\mathsf{out}}_j \\
\text{declarations} & D & ::= & \kw{fn}\ f(\kw{\&mut}\,x_1, \ldots,
  \kw{\&mut}\,x_n, y_1, \ldots, y_m) : \mathit{Sig}\ \{\, e \,\} \\
\text{contexts} & \Gamma & ::= & \cdot \mid \Gamma, x : \tau
\end{array}
\]
\caption{Syntax of $\lambda$-Forte. $p$ ranges over the trusted
primitives $\Sigma$ and $f$ over the declared functions $\Phi$; both have
signatures. In a signature the sensitive types draw their environments
from the variables $\vec{\alpha}$, and each output type
$\tau^{\mathsf{out}}_j$ is the $j$th borrowed parameter's type with its
environment replaced by an expression $g_j(\vec{\alpha})$ over those
variables.}
\Description{Grammar of values, expressions, arguments, types, signatures, declarations, and contexts.}
\label{fig:syntax}
\end{figure}

\paragraph{Sources and environments.}
A program declares a finite set of $k$ \emph{sources}, the units of
protection. Each source $i$ has a metric $M_i$ on its input space: the
absolute difference for scalar sources, and for dataset sources either the
elementwise $L_1$ distance or the substituted-row count on sized data. A
\emph{sensitivity environment} $\sens{\ensuremath{\vec{s}} \in
\mathbb{R}_{\ge 0}^{k}}$ assigns one coefficient per source; the
environment attached to a value bounds how far that value can move, in the
value's own metric, when the sources move: by at most
$\sens{\vec{s}}\cdot\vec{d}$ when source $i$ moves by $d_i$ in $M_i$.

\paragraph{Types.}
Types are public base types $\mathsf{pub}\,B$, sensitive scalars
$S[\sens{\vec{s}}]$, and sensitive vectors under three metrics: the
elementwise $L_1$ distance, the substituted-row count, and the $L_\infty$
distance. The $\mathsf{Sub}$ forms are for sized data; the bounded form
additionally records the domain invariant $[lo,hi]$ that licenses the
diameter sum rule (\S\ref{sec:overview}), and the $L_\infty$ form types the
count and score vectors that metric-converting primitives produce.
Subtyping raises environments pointwise: environments are upper bounds,
so weakening is sound. \rulename{T-Shared} types a shared borrow $\&x$ as
a read of $x$; a shared reference is immutable for its lifetime, so the
calculus treats it as a value.

\paragraph{Judgment.}
The typing judgment is flow-sensitive, with an output context recording the
types after evaluation:
\[
\Gamma \vdash e : \tau \dashv \Gamma'.
\]
Assignment retypes the assigned local, so ordinary mutation of locals is
already a strong update, as in the imperative sensitivity systems
(\S\ref{sec:related}). The rules that are new govern calls that mutate
through borrows, at a primitive and across a function boundary, and loops.

\begin{figure}
\small
\begin{mathpar}
\inferrule*[right=T-Var]
  {\Gamma(x) = \tau}
  {\Gamma \vdash x : \tau \dashv \Gamma}

\inferrule*[right=T-Shared]
  {\Gamma(x) = \tau}
  {\Gamma \vdash \&x : \tau \dashv \Gamma}

\inferrule*[right=T-Lit]
  {v \text{ a literal of base type } B}
  {\Gamma \vdash v : \mathsf{pub}\,B \dashv \Gamma}

\inferrule*[right=T-Arith]
  {\Gamma \vdash e_1 : \mathsf{pub}\,B \dashv \Gamma_1 \\
   \Gamma_1 \vdash e_2 : \mathsf{pub}\,B \dashv \Gamma_2}
  {\Gamma \vdash e_1 \oplus e_2 : \mathsf{pub}\,B \dashv \Gamma_2}

\inferrule*[right=T-Let]
  {\Gamma \vdash e_1 : \tau_1 \dashv \Gamma_1 \\
   \Gamma_1, x{:}\tau_1 \vdash e_2 : \tau \dashv \Gamma_2, x{:}\tau_1'}
  {\Gamma \vdash \kw{let}\ x = e_1\ \kw{in}\ e_2 : \tau \dashv \Gamma_2}

\inferrule*[right=T-Assign]
  {\Gamma \vdash e : \tau \dashv \Gamma'}
  {\Gamma \vdash x := e : \mathsf{unit} \dashv \Gamma'[x \mapsto \tau]}

\inferrule*[right=T-Seq]
  {\Gamma \vdash e_1 : \tau_1 \dashv \Gamma_1 \\
   \Gamma_1 \vdash e_2 : \tau \dashv \Gamma_2}
  {\Gamma \vdash e_1;\, e_2 : \tau \dashv \Gamma_2}

\inferrule*[right=T-Sub]
  {\Gamma \vdash e : \tau \dashv \Gamma' \\
   \tau <: \tau' \\ \Gamma' \sqsubseteq \Gamma''}
  {\Gamma \vdash e : \tau' \dashv \Gamma''}

\inferrule*[right=T-If]
  {\Gamma \vdash e : \mathsf{pub}\,\mathsf{bool} \dashv \Gamma_0 \\
   \Gamma_0 \vdash e_1 : \tau \dashv \Gamma_1 \\
   \Gamma_0 \vdash e_2 : \tau \dashv \Gamma_2 \\
   \Gamma_1 \sqsubseteq \Gamma' \\ \Gamma_2 \sqsubseteq \Gamma'}
  {\Gamma \vdash \kw{if}\ e\ \kw{then}\ e_1\ \kw{else}\ e_2 : \tau
   \dashv \Gamma'}

\inferrule*[right=T-While]
  {\Gamma \sqsubseteq \Gamma_{\mathit{inv}} \\
   \Gamma_{\mathit{inv}} \vdash e_c : \mathsf{pub}\ \mathsf{bool} \dashv
   \Gamma_{\mathit{inv}} \\
   \Gamma_{\mathit{inv}} \vdash e_b : \tau \dashv \Gamma_{\mathit{inv}}}
  {\Gamma \vdash \kw{while}\ e_c\ \kw{do}\ e_b : \mathsf{unit} \dashv
   \Gamma_{\mathit{inv}}}

\inferrule*[right=T-CallPure]
  {\Sigma(p) = \forall\vec{\alpha}.\ (\tau_1, \ldots, \tau_m) \to \tau \\
   \sigma \text{ instantiates } \vec{\alpha} \\
   \Gamma \vdash a_i : \sigma\tau_i \dashv \Gamma \quad (i \le m)}
  {\Gamma \vdash p(a_1, \ldots, a_m) : \sigma\tau \dashv \Gamma}

\inferrule*[right=T-CallMut]
  {\Sigma(p) = \forall\vec{\alpha}.\ (\kw{\&mut}\,\tau_1, \ldots,
     \kw{\&mut}\,\tau_n, \tau'_1, \ldots, \tau'_m) \to \tau\
     \mathsf{ensures}\ x_j : \tau^{\mathsf{out}}_j \\
   \sigma \text{ instantiates } \vec{\alpha} \\
   \Gamma(z_j) = \sigma\tau_j \quad (j \le n) \\
   z_1 \ldots z_n \ \text{pairwise distinct} \\
   \Gamma \vdash a_i : \sigma\tau'_i \dashv \Gamma \quad (i \le m)}
  {\Gamma \vdash p(\kw{\&mut}\,z_1, \ldots, \kw{\&mut}\,z_n,
     a_1, \ldots, a_m) : \sigma\tau
   \dashv \Gamma[z_j \mapsto \sigma\tau^{\mathsf{out}}_j]}

\inferrule*[right=T-CallFn]
  {\Phi(f) = \forall\vec{\alpha}.\ (\kw{\&mut}\,\tau_1, \ldots,
     \kw{\&mut}\,\tau_n, \tau'_1, \ldots, \tau'_m) \to \tau\
     \mathsf{ensures}\ x_j : \tau^{\mathsf{out}}_j \\
   \sigma \text{ instantiates } \vec{\alpha} \\
   \Gamma(z_j) = \sigma\tau_j \quad (j \le n) \\
   z_1 \ldots z_n \ \text{pairwise distinct} \\
   \Gamma \vdash a_i : \sigma\tau'_i \dashv \Gamma \quad (i \le m)}
  {\Gamma \vdash f(\kw{\&mut}\,z_1, \ldots, \kw{\&mut}\,z_n,
     a_1, \ldots, a_m) : \sigma\tau
   \dashv \Gamma[z_j \mapsto \sigma\tau^{\mathsf{out}}_j]}

\inferrule*[right=WF-Fn]
  {\Phi(f) = \forall\vec{\alpha}.\ (\kw{\&mut}\,\tau_1, \ldots,
     \kw{\&mut}\,\tau_n, \tau'_1, \ldots, \tau'_m) \to \tau\
     \mathsf{ensures}\ x_j : \tau^{\mathsf{out}}_j \\
   \text{for every ground instantiation } \sigma \text{ of } \vec{\alpha}: \\
   x_1{:}\sigma\tau_1, \ldots, x_n{:}\sigma\tau_n,
   y_1{:}\sigma\tau'_1, \ldots, y_m{:}\sigma\tau'_m
   \vdash e : \sigma\tau \dashv \Gamma' \\
   \Gamma'(x_j) <: \sigma\tau^{\mathsf{out}}_j \quad (j \le n)}
  {\vdash \kw{fn}\ f(\kw{\&mut}\,x_1, \ldots, \kw{\&mut}\,x_n,
     y_1, \ldots, y_m) : \Phi(f)\ \{\, e \,\}\ \mathsf{ok}}
\end{mathpar}
\caption{Typing rules of $\lambda$-Forte. Subtyping $\tau <: \tau'$ raises
environments pointwise and is the identity on public types;
$\Gamma \sqsubseteq \Gamma'$ is its pointwise extension to contexts. A
program is well-typed when every declaration is \textsf{ok} and its main
expression types in the empty context.}
\Description{The typing rules as inference rules, including the two call rules with a distinctness premise and the function well-formedness rule.}
\label{fig:typing}
\end{figure}

\begin{figure}
\small
\begin{mathpar}
\inferrule*[right=E-Assign]
  {\langle e, \sigma, \varsigma \rangle \Downarrow \langle v, \sigma' \rangle}
  {\langle x := e, \sigma, \varsigma \rangle \Downarrow
   \langle (), \sigma'[x \mapsto v] \rangle}

\inferrule*[right=E-WhileF]
  {\langle e_c, \sigma, \varsigma \rangle \Downarrow
   \langle \kw{false}, \sigma \rangle}
  {\langle \kw{while}\ e_c\ \kw{do}\ e_b, \sigma, \varsigma \rangle
   \Downarrow \langle (), \sigma \rangle}

\inferrule*[right=E-WhileT]
  {\langle e_c, \sigma, \varsigma \rangle \Downarrow
   \langle \kw{true}, \sigma \rangle \\
   \langle e_b, \sigma, \varsigma \rangle \Downarrow \langle v, \sigma_1 \rangle \\
   \langle \kw{while}\ e_c\ \kw{do}\ e_b, \sigma_1, \varsigma \rangle
   \Downarrow \langle (), \sigma_2 \rangle}
  {\langle \kw{while}\ e_c\ \kw{do}\ e_b, \sigma, \varsigma \rangle
   \Downarrow \langle (), \sigma_2 \rangle}

\inferrule*[right=E-CallMut]
  {\langle a_i, \sigma, \varsigma \rangle \Downarrow \langle v_i, \sigma \rangle
   \quad (i \le m) \\
   \llbracket p \rrbracket(\sigma(z_1), \ldots, \sigma(z_n), v_1, \ldots, v_m)
   = (w_1, \ldots, w_n, v)}
  {\langle p(\kw{\&mut}\,z_1, \ldots, \kw{\&mut}\,z_n, a_1, \ldots, a_m),
   \sigma, \varsigma \rangle \Downarrow
   \langle v, \sigma[z_1 \mapsto w_1, \ldots, z_n \mapsto w_n] \rangle}

\inferrule*[right=E-CallFn]
  {\kw{fn}\ f(\kw{\&mut}\,x_1, \ldots, \kw{\&mut}\,x_n, y_1, \ldots, y_m)
   \ \{\, e_f \,\} \in \Phi \\
   \langle a_i, \sigma, \varsigma \rangle \Downarrow \langle v_i, \sigma \rangle
   \quad (i \le m) \\
   u_1, \ldots, u_m \text{ fresh} \\
   \langle e_f[x_j := z_j][y_i := u_i],
     \sigma[u_i \mapsto v_i], \varsigma \rangle
   \Downarrow \langle v, \sigma' \rangle}
  {\langle f(\kw{\&mut}\,z_1, \ldots, \kw{\&mut}\,z_n, a_1, \ldots, a_m),
   \sigma, \varsigma \rangle \Downarrow
   \langle v, \sigma' \setminus \{u_1, \ldots, u_m\} \rangle}
\end{mathpar}
\caption{Evaluation rules of $\lambda$-Forte that write the store. Evaluation
is big-step and deterministic over a local store $\sigma$ and a source
store $\varsigma$; the omitted rules (literals, variables, arithmetic,
\kw{let}, sequencing, conditionals, pure calls) are standard and write
nothing. A primitive $p$ denotes a function $\llbracket p \rrbracket$ on
values that returns a result and one post-value per borrowed parameter; it
writes the borrowed locals and no other. A function call binds each
borrowed parameter to the caller's local itself, by substitution, and each
value parameter to a fresh local, and renames the callee's own
\kw{let}-bound locals apart from the caller's.}
\Description{Big-step evaluation rules for assignment, loops, primitive calls, and function calls.}
\label{fig:semantics}
\end{figure}

\paragraph{Strong updates through exclusive borrows.}
\rulename{T-CallMut} types a call to a primitive whose signature has an
\kw{ensures} clause: after the call, each borrowed local $z_j$ has the
$j$th output type of the signature, instantiated at the actuals'
environments. Pairwise distinctness of the borrowed locals is the
calculus-level shadow of Rust's exclusivity. The borrow checker discharges
it in the implementation; \S\ref{sec:metatheory} shows it is the premise
soundness needs. The by-value
fragment needs the premise only across the borrowed arguments themselves:
repeating a local as a \emph{value} argument, as in
$\kw{add\_assign}(\kw{\&mut}\,x, x)$, is sound, because the call reads the
value out before the write and the output environment
$\sens{\vec{s}+\vec{s}}$ correctly prices the double read. Index
arithmetic, in other words, charges contraction; only colliding write
capabilities need exclusivity.

\paragraph{Borrows across a function boundary.}
\rulename{WF-Fn} checks a declared function once, under its signature
scheme: it types the body with each borrowed parameter standing for a
distinct location at the environment the scheme names, and it requires
the output context to weaken into the \kw{ensures} types. A call,
\rulename{T-CallFn}, instantiates the scheme at the actuals' environments
and reuses the signature. The rule has the shape of \rulename{T-CallMut},
and its distinctness premise now does more than shadow the borrow
checker. Evaluation binds the callee's parameters to the caller's locals
by substitution (\rulename{E-CallFn}), and the callee's typing derivation
transports to the call site only when that substitution is injective on
the borrowed parameters, which is the distinctness of the actuals. With
two actuals equal, the callee's derivation describes two locations where
the call supplies one, and the \kw{ensures} entry for the second
parameter is a claim about a location the first parameter has updated;
\S\ref{sec:metatheory} exhibits the resulting false claim.
$\lambda$-Forte therefore places its demand on the host at the function
boundary: the checker may check the callee in isolation because every
call site must supply distinct locations, and Rust's borrow checker
discharges that obligation at every call.

\paragraph{Loops.}
A loop has an invariant context $\Gamma_{\mathit{inv}}$, which the
calculus quantifies existentially and the implementation infers
(\S\ref{sec:implementation}). Guards are public. Branching on sensitive
data would reintroduce the metric-blowup analysis of Fuzz-style case
elimination and is out of scope for the core; the case studies never
require it at the checked layer (\S\ref{sec:casestudies}).

\paragraph{The signature set.}
$\Sigma$ mirrors the trusted API of the implementation: constants and
source injections, \kw{add}, \kw{sub}, non-negative \kw{scale}, the
Lipschitz \kw{clamp}, the two vector sum rules ($L_1$-preserving on
$\mathsf{Vec}_{L1}$; diameter-scaling on the bounded $\mathsf{Sub}$ form,
with no rule at all on the unbounded form), the metric converters
(\kw{histogram}, \kw{prefix\_counts}, \kw{prefix\_sums}), and the
structural operations (\kw{chunk\_sums}, \kw{concat}, \kw{add\_at},
\kw{dup}). Each entry has a semantic validity obligation, which the leaf theorems
of \S\ref{sec:implementation} discharge once, in mechanized form;
\S\ref{sec:metatheory} consumes only their statements.

\section{Metatheory}
\label{sec:metatheory}

Soundness for $\lambda$-Forte is metric preservation, the property the
Fuzz--Solo lineage proves over substitution-based semantics, restated over
a store. Fix distances $\vec{d} \in \mathbb{R}_{\ge 0}^{k}$.

\begin{definition}[Relatedness]
\label{def:rel}
Source stores are related, $\varsigma_1 \sim_{\vec{d}} \varsigma_2$, when
source $i$'s inputs are within $d_i$ in $M_i$ for every $i$. Values are
related at a type as follows: at $\mathsf{pub}\,B$, equality; at
$S[\sens{\vec{s}}]$, $|v_1 - v_2| \le \sens{\vec{s}}\cdot\vec{d}$; at the
vector types, equal lengths and the same bound with the $L_1$ distance,
the substituted-row count, or the $L_\infty$ distance in place of
$|\cdot|$, with the bounded $\mathsf{Sub}$ form additionally requiring all
entries in $[lo,hi]$. Stores are related at a context,
$\sigma_1 \approx_{\Gamma} \sigma_2$, when every $x \in \mathsf{dom}(\Gamma)$
is related at $\Gamma(x)$.
\end{definition}

\begin{theorem}[Metric Preservation]
\label{thm:mp}
Let every declaration of the program be \textsf{ok}. If
$\Gamma \vdash e : \tau \dashv \Gamma'$,
$\varsigma_1 \sim_{\vec{d}} \varsigma_2$,
$\sigma_1 \approx_{\Gamma} \sigma_2$, and
$\langle e, \sigma_1, \varsigma_1 \rangle \Downarrow
\langle v_1, \sigma_1' \rangle$, then
$\langle e, \sigma_2, \varsigma_2 \rangle \Downarrow
\langle v_2, \sigma_2' \rangle$ for some $v_2$ and $\sigma_2'$ with
$v_1 \approx_{\tau} v_2$ and $\sigma_1' \approx_{\Gamma'} \sigma_2'$.
\end{theorem}

\begin{proof}
By induction on the first run's evaluation derivation, inverting the typing
derivation; public guards are equal across related runs, so the second run
follows the same rule skeleton and terminates with the first. The full case
analysis is in the appendix; the lemmas below cover the cases that are new
relative to the pure lineage.
\end{proof}

The theorem's conclusion includes the second run's termination. Guards are
public, so related stores drive both runs through equal guard values, and
one run diverges on a related input only when the other does. The
theorem therefore extracts no sensitivity claim from a program that fails
to terminate.

\begin{lemma}[Leaf validity]
\label{lem:leaf}
Every signature in $\Sigma$ is semantically valid: related arguments
produce related results and related post-values for borrowed parameters.
\end{lemma}

Lemma~\ref{lem:leaf} is the single seam between the type system and
primitive behavior, and a machine-checked theorem discharges every
non-definitional case (\S\ref{sec:implementation}). For example, the
diameter sum rule's case states that on same-length runs with entries in
$[lo,hi]$, the clamped sums differ by at most the substituted-row count
times $hi - lo$.

\begin{lemma}[Frame]
\label{lem:frame}
A primitive call writes the locations of its borrowed arguments and no
other. A function call writes the locations of its borrowed actuals and
the callee's fresh locals, which scope out, and no other. The call
preserves relatedness at every other variable of the context verbatim.
\end{lemma}

\begin{lemma}[Strong update]
\label{lem:update}
Under \rulename{T-CallMut}'s premises, the updated stores are related at
the updated context: Lemma~\ref{lem:leaf} relates them at each borrowed
$z_j$, and Lemma~\ref{lem:frame} elsewhere. The proof consumes pairwise
distinctness here: it guarantees that each location receives one
post-value, so the pointwise check at $z_j$ observes the value
Lemma~\ref{lem:leaf} describes.
\end{lemma}

\begin{lemma}[Renaming]
\label{lem:rename}
Typing and evaluation are invariant under an injective renaming of locals:
if $\rho$ is injective on the locals of $e$, then
$\Gamma \vdash e : \tau \dashv \Gamma'$ if and only if
$\rho\Gamma \vdash \rho e : \tau \dashv \rho\Gamma'$, and
$\langle e, \sigma, \varsigma \rangle \Downarrow \langle v, \sigma' \rangle$
if and only if
$\langle \rho e, \rho\sigma, \varsigma \rangle \Downarrow
\langle v, \rho\sigma' \rangle$.
\end{lemma}

\begin{lemma}[Function call]
\label{lem:call}
Under \rulename{T-CallFn}'s premises, with $f$ \textsf{ok}, the call's
output stores are related at the output context. Distinctness of the
actuals makes the substitution $[x_j := z_j]$ injective, so
Lemma~\ref{lem:rename} transports the callee's derivation at the
instantiation $\sigma$ to the call site; the induction hypothesis on the
body relates the callee's post-stores at its output context, which weakens
into the \kw{ensures} types; Lemma~\ref{lem:frame} covers every other
variable of the caller.
\end{lemma}

\paragraph{Necessity of exclusivity.}
The theorem needs the distinctness premise as a hypothesis. Declare
\[
\begin{array}{l}
\kw{fn}\ \kw{double\_into}(p\colon \kw{\&mut}\,S[\sens{\alpha}],\;
q\colon \kw{\&mut}\,S[\sens{\beta}]) \\
\quad \mathsf{ensures}\ p\colon S[\sens{\alpha+\beta}],\;
q\colon S[\sens{\beta}] \ \{\, p := p + q \,\},
\end{array}
\]
which is \textsf{ok}: its body adds $q$ into $p$ and leaves $q$ unchanged.
Applied to two borrows of one local $x$ with $\Gamma(x) = S[\sens{1}]$,
the output context retains, through the $q$ entry, the claim
$x : S[\sens{1}]$, while the location now holds a value that two runs at
source distance 1 place at distance 2. A subsequent read of $x$ propagates
the false claim. The failure is a frame violation: the rule treated the
two \kw{ensures} entries as claims about distinct locations. The update
arithmetic itself tolerates aliasing, as noted in \S\ref{sec:calculus}, so
the proof consumes exclusivity where Lemmas~\ref{lem:update}
and~\ref{lem:call} used distinctness, and nowhere else.

\paragraph{No linearity.}
Weakening, contraction, and exchange are all admissible in the sensitivity
layer; environment addition prices duplication of claims at use sites. The budget layer alone is affine, and by library move semantics
in place of any rule of the calculus. This locates the demand
$\lambda$-Forte places on its host: exclusivity of write capabilities, and
nothing substructural about data.

\paragraph{Mechanization.}
Verus mechanizes the scalar fragment of the calculus: syntax, a
fuel-based interpreter, function-style flow-sensitive typing with annotated
join and invariant contexts, and Theorem~\ref{thm:mp} with the second
run's termination folded into its conclusion, in 46 verified items with no
assumptions beyond the solver. The fragment comprises mutable locals with
assignment as strong update, public conditionals and loops, source reads,
the scalar primitives, \kw{add\_assign} through one borrow with a by-value
argument, and declared procedures whose parameters are all exclusive
borrows. Evaluation goes through an \emph{access map} from variables to
store cells; a call binds the callee's parameters to the caller's cells,
so the call is by reference, and the theorem's hypotheses include the
injectivity of the access map, which the call case establishes from the
distinctness of the actuals; no other case uses it. The mechanized
procedure rule is \rulename{T-CallFn} for monomorphic signatures and
bodies that mention only their parameters; Lemma~\ref{lem:rename} handles
environment-polymorphic schemes, by-value parameters, and callee-fresh
locals on paper. The development mechanizes the refutation as well: the
theorem \kw{aliased\_call\_refutation} declares \kw{double\_into}, proves
its signature valid, and exhibits related input stores whose outputs under
the aliased call are unrelated at the premise-free output context, while
the typing function rejects the call. The library's development
mechanizes Lemma~\ref{lem:leaf}'s vector cases separately
(\S\ref{sec:implementation}).

\section{Implementation}
\label{sec:implementation}

\kw{forte-rs} realizes Forte as an ordinary Rust library, and
Flux~\cite{flux} checks it, with no fork of the compiler and no new
checker. A
sensitivity-bearing type is an \kw{\#[opaque]} struct whose refinement
index is the environment; the payload is invisible to the checker, which
is Solo's split between a runtime value and its ghost accounting. For a
fixed program the source set is finite, so an environment is a $k$-tuple
index, with no map needed, and all environment arithmetic in signatures is
plain SMT arithmetic. The width is fixed because a refinement index has
one sort and the tuple's arity is part of that sort; the library
instantiates it at three, which every kernel fits, and another width
means regenerating the types. A map-sorted index would lift the
restriction in principle, but the rules need addition, scaling, and
comparison of whole environments, which SMT arrays do not supply without
quantifiers, whereas the tuple keeps that arithmetic linear and decidable,
and that linearity lets Flux infer the loop invariants. Flux's generic
refinements~\cite{genericrefinements} abstract over a predicate on an
index of fixed sort, which suits polymorphism over metrics
(\S\ref{sec:byexample}); library-level source polymorphism turns on
whether they can abstract over the width of the environment itself. Strong updates come from Flux's existing treatment
of \kw{\&mut} with an \kw{ensures} clause, at trusted signatures and at
checked functions alike; \rulename{T-CallMut} and \rulename{T-CallFn}
required no extension of Flux. Table~\ref{tab:encoding} lists the
encoding.

\begin{table}
\centering
\small
\begin{tabular}{LL}
\toprule
$\lambda$-Forte & \kw{forte-rs} \\
\midrule
$S[\sens{\vec{s}}]$, $\mathsf{Vec}_{\cdot}[\sens{\vec{s}}]$ &
  \kw{\#[opaque]} structs refined by a 3-tuple of \kw{int} or \kw{real} sort \\
$\Sigma$, with Lemma~\ref{lem:leaf} obligations &
  \kw{\#[trusted]} signatures, each naming its Verus theorem \\
$\Phi$, under \rulename{WF-Fn} &
  checked functions with \kw{spec} signatures \\
$\forall\vec{\alpha}$ in signatures &
  refinement parameters (\kw{@s}) \\
\rulename{T-CallMut}, \rulename{T-CallFn} &
  \kw{\&mut} parameters with \kw{ensures} clauses \\
distinctness of borrowed actuals &
  the borrow checker of \kw{rustc} \\
\rulename{T-Sub} &
  implication between index refinements, which SMT discharges \\
\rulename{T-While} with $\Gamma_{\mathit{inv}}$ &
  inferred loop invariants \\
$\mathsf{pub}\,B$ &
  unrefined \kw{i64}, \kw{usize}, \kw{bool}, and refined public values \\
\bottomrule
\end{tabular}
\caption{The encoding of $\lambda$-Forte in \kw{forte-rs}.}
\label{tab:encoding}
\end{table}

\paragraph{Unary and relational.}
Flux's soundness theorem is unary: an accepted program satisfies its
refinements on every single run. A sensitivity claim relates two runs, and
no unary statement about an opaque index entails it. The host therefore
cannot supply metric preservation, and the encoding claims that the host's
unary discipline transports it. We make that claim a theorem. Fix the \emph{fragment} $F$ of \kw{forte-rs} programs: bodies in
the syntax of \S\ref{sec:calculus} whose sensitive-typed refinements are
conjunctions of per-coordinate atoms $\nu.i \le e$, $\nu.i \ge e$, or
$\nu.i = e$ with $e$ an expression over public values, at least one upper
bound per coordinate. Every kernel of \S\ref{sec:casestudies} is in $F$.
The \emph{elaboration} of such a refinement is the environment whose $i$th
coordinate is the least upper bound the atoms give $\nu.i$; the
elaboration of a checked signature is the scheme whose environment
variables are the signature's refinement parameters.

\begin{theorem}[Correspondence]
\label{thm:corr}
Let $P$ be a program in $F$ that the checker accepts, under the model of
the check that Appendix~\ref{app:corr} states. Then every function of $P$
is \textsf{ok} in $\lambda$-Forte at the elaboration of its signature,
and the elaboration of the type the checker assigns at any program point
is a $\lambda$-Forte type for it. Consequently Theorem~\ref{thm:mp}
applies to $P$.
\end{theorem}

The model of the check is a judgment with four properties, each a
documented feature of Flux: only signatures and subtyping constrain
refinements on opaque types, and subtyping is implication that the solver
discharges; the checker updates \kw{\&mut} arguments to signatures with
\kw{ensures} clauses in place, and \kw{rustc} rejects two live exclusive
borrows of one local; the checker checks loops against an inferred
invariant and joins conditionals by a common weakening; refinements on
public values are unary facts that hold on both runs, since public values
are equal across them. The model is the restriction of the core calculus
of the Flux paper~\cite{flux} to this fragment, with the opaque structs
and their signatures as the only constructs that produce a
sensitive-typed value; deriving the four properties from that calculus,
in place of assuming them, is future work. The proof
(Appendix~\ref{app:corr}) is an induction on the model
derivation in which each rule maps to the $\lambda$-Forte rule of the same
name, with one lemma doing the work: implication between refinements in
$F$ elaborates to pointwise environment inequality, because a satisfiable
conjunction of per-coordinate atoms contains its own coordinatewise
maximum. The paragraph on the trusted computing base below states what
the theorem leaves to trust.

\paragraph{Two index sorts.}
Flux's \kw{real} and \kw{int} refinement sorts do not unify, and its
refinement language has no cast between them, so no constraint can
equate a \kw{real}-sorted environment with a \kw{usize} loop count. The library
therefore splits environments by purpose: \kw{real}-sorted indices for
literal-constant composition, where fractional coefficients cost nothing,
and \kw{int}-sorted indices (scaled rationals) wherever a bound mentions a
loop count or a runtime value, since \kw{i64} values lift into \kw{int}
indices directly. Every case-study kernel uses the \kw{int} sort; the
library's own regression loop exercises the \kw{real} sort. The same
boundary rules out stating $L_2$ claims: the refinement grammar has no
square root, so the grammar confines $L_2$ claims to radius-style
over-approximation at leaves, as in the sum-of-squares rule below. It also
keeps a sized mean at the level of its components: a bound of the form
$(hi - lo)/n$ is a rational in a runtime value, which no \kw{int}-sorted
index expresses, so the checked kernel releases the clamped sum and the
insensitive count and leaves the public division past the sensitivity
boundary (\S\ref{sec:casestudies}).

\paragraph{The trusted base, machine-checked.}
Every \kw{\#[trusted]} signature with a two-run obligation names a
Verus~\cite{verus} theorem that discharges its Lemma~\ref{lem:leaf}
obligation, in a satellite development of 43 verified theorems and
supporting lemmas. The correspondence is one-to-one, and each signature's
documentation comment records it; representative rows:

\begin{table}[t]
\centering
\begin{tabular}{LL}
\toprule
Trusted signature & Leaf theorem \\
\midrule
\kw{SInt::add}, \kw{add\_assign} & \kw{add\_sensitivity} \\
\kw{SInt::clamp} & \kw{clamp\_lipschitz} \\
\kw{SVecBnd::sum} (diameter rule) & \kw{sum\_clamped\_k\_subst\_diameter} \\
\kw{SVecBnd::sum\_squares} & \kw{sum\_squares\_bounded\_k\_subst} \\
\kw{SVec::histogram} & \kw{histogram\_subst\_l1} \\
\kw{SVec::prefix\_counts} & \kw{count\_le\_subst\_stable} \\
\kw{SVecL1::chunk\_sums} & \kw{chunk\_sums\_l1\_stable} \\
\kw{SVecL1::add\_at} & \kw{add\_at\_l1} \\
\bottomrule
\end{tabular}
\caption{The trusted-base seam (excerpt). Each \kw{\#[trusted]} Flux
signature cites the Verus theorem that proves its two-run validity.}
\label{tab:leaves}
\end{table}

This addresses the structural weakness of the lineage noted in
\S\ref{sec:intro}: a trusted primitive signature outside the reach of the
metatheory. Here the leaf layer proves the trusted signatures as stated,
and the calculus's Lemma~\ref{lem:leaf} consumes them by name.

\paragraph{Inference.}
The annotation burden of the checked layer is small and concentrated at
loops. Flux infers loop invariants from qualifier shapes it scrapes from
signatures; with scraping enabled (one project-wide switch), every
straight-line kernel and the accumulation loops of \S\ref{sec:overview}
verify with no annotations beyond their signatures. The b-ary tree kernel
is the sole exception observed: its invariant couples the accumulated
environment to the loop counter ($\kw{out}.a \le i$), a shape absent from
its signature, and it needs 2 qualifier hints, which ablation over the
hint subsets verifies minimal. The negative direction is equally
important: the checker rejects a 1-sensitivity claim on the clamped sum, a
tightened loop bound, and an off-by-one diameter constant, so the
inference is doing arithmetic and no pattern-matching would do in its
place.

\paragraph{Budgets.}
Privacy budgets are affine tokens: a non-\kw{Copy} \priv{\kw{Budget}}
indexed by a scaled $\priv{\varepsilon}$, with \kw{split} and \kw{spend}
consuming their argument by move. Rust's move semantics supplies the
affine discipline; we add no substructural typing. A release primitive
closes the pipeline end to end: \kw{laplace\_release} consumes the
sensitive value and a budget token, with the Laplace calibration
$\priv{\varepsilon} \ge \sens{\Delta}/b$ as an integer side condition on
the signature (\kw{eps} $\cdot$ \kw{scale} $\ge 10^{6} \cdot$
\kw{sensitivity}, in micro-epsilons), so an under-budgeted or
under-calibrated release is a type error.

\paragraph{Trusted computing base.}
The guarantee rests on the following, in decreasing order of the evidence
behind each. The two-run validity of every deterministic trusted
signature is a Verus theorem, and Verus mechanizes Theorem~\ref{thm:mp}
for the scalar fragment; we trust Verus and its solver. We trust the Rust
body of each trusted primitive to implement the function the
corresponding theorem describes, by inspection, one short body per
signature, 313 non-comment lines for the whole library outside the
kernels. We trust Flux, \kw{rustc}, the fixpoint engine, and the SMT
solver to implement the model of Theorem~\ref{thm:corr}. The declaration
of sources, the metric the user gives each source, and the injection of
raw data at a source (\kw{from\_source}) are the user's assertions about
the data. Of the probabilistic layer, the privacy-loss bound behind the
release primitive's side condition and its sequential composition are
Verus theorems (\S\ref{sec:byexample}); we trust the Laplace sampler to
draw from the discrete Laplace density, and the step from a loss bound to
a density-ratio bound is the monotonicity of $\exp$. The budget tokens make splitting and spending affine, but nothing
prevents a program from minting a second token, so one mint per program
run is a protocol the checker does not enforce.

\section{Case studies}
\label{sec:casestudies}

We evaluated \forte{} by re-implementing the kernels of 5 OpenDP
transformation families, selected for genuine in-place mutation, and
checking the certified constants against OpenDP's stability maps and proof
documents as ground truth. Every kernel composes primitives all the way
through: no \kw{\#[trusted]} annotation appears in any kernel, and the leaf
theorems of \S\ref{sec:implementation} cover every primitive the kernels
use. Table~\ref{tab:kernels} summarizes the
comparison of constants and Table~\ref{tab:burden} the annotation burden
and checking cost.

\begin{table}
\centering
\small
\begin{tabular}{LCC}
\toprule
Kernel & \multicolumn{1}{C}{OpenDP stability map} &
\multicolumn{1}{C}{\forte{} certified} \\
\midrule
clamp $\circ$ sized sum & $\lfloor d/2 \rfloor (U{-}L)$ & $\sens{hi{-}lo}$ per row \\
quantile prefix counts (fused) & $\lfloor d/2 \rfloor$ per candidate & $\sens{1}$ per row \\
quantile prefix counts (composed) & \multicolumn{1}{C}{--} & $\sens{2}$ per row \\
b-ary tree & $\mathit{num\_layers}$ (trusted) & $\sens{\kw{out}.a \le \mathit{layers}}$ (inferred) \\
count & $d$ (unsized, add/remove) & $\sens{0}$ (sized, substitution) \\
mean components & $(U{-}L)/n$ & $\sens{hi{-}lo}$ (sum), $\sens{0}$ (count) \\
sum of squares & (deviations analogue) & $\sens{hi^2{-}lo^2}$ per row \\
count by categories & $d$ ($L_1$, symmetric) & $\sens{2}$ per row \\
\bottomrule
\end{tabular}
\caption{Kernels versus OpenDP ground truth. OpenDP's $d$ is symmetric
distance, 2 per substituted row on sized data. The quantile rows certify
the prefix-count vector; OpenDP's map additionally scales by a public
factor from the quantile's $\alpha$, elided here. The count row compares
two metrics, as \S\ref{sec:casestudies} discusses; the mean row stops at
the components because the public division is outside the \kw{int} sort
(\S\ref{sec:implementation}).}
\label{tab:kernels}
\end{table}

\begin{table}
\centering
\small
\begin{tabular}{@{}LLNNNNN@{}}
\toprule
Module & Kernels &
\multicolumn{1}{R}{Lines} & \multicolumn{1}{R}{Spec} &
\multicolumn{1}{R}{Sigs} & \multicolumn{1}{R}{Hints} &
\multicolumn{1}{R}{Check (ms)} \\
\midrule
\kw{sum\_bounded} & \cellstack{clamp $\circ$ sum, helper,\\ 3 sources, $n$-day loop} & 45 & 11 & 5 & 0 & 51 \\
\kw{quantile} & fused and composed routes & 12 & 2 & 2 & 0 & 18 \\
\kw{b\_ary\_tree} & layer accumulation & 17 & 2 & 1 & 2 & 110 \\
\kw{stats} & count, mean, squares, categories & 23 & 5 & 4 & 0 & 19 \\
\kw{weekday} & element writes at public indices & 17 & 1 & 1 & 0 & 38 \\
\kw{dp\_release} & release, budget split & 17 & 2 & 2 & 0 & 20 \\
\midrule
total & & 131 & 23 & 15 & 2 & \\
\bottomrule
\end{tabular}
\caption{Annotation burden and checking cost by kernel module. ``Lines''
counts non-blank, non-comment lines; ``Spec'' the lines of those that are
signatures; ``Sigs'' the number of signatures; ``Hints'' the qualifier
shapes beyond the signatures (Flux infers all other invariants); and
``Check'' the Flux time for the slowest function in the module, from the
checker's per-function profile. The whole crate, 64 functions, checks in
0.41\,s; the 43 leaf theorems verify in 0.68\,s and the calculus
development, 46 items, in 1.31\,s. Flux at commit \kw{a166721} (August
2026) with z3 5.1 and the fixpoint solver, and Verus 0.2026.08.23, on an
Apple M4 laptop with 32\,GB of memory.}
\label{tab:burden}
\end{table}

\paragraph{Constants and conventions.}
Every certified constant agrees with OpenDP's after one unit conversion,
stated once: OpenDP's $d$ counts symmetric distance, which is 2 per
substitution on sized data, so its $\lfloor d/2 \rfloor$ prefixes
correspond to our per-substituted-row accounting. Agreement is the
expected outcome: the constants that OpenDP justifies by proof documents,
comments, and citations are here subtyping obligations that the solver
discharges on every call.

\paragraph{Across a function boundary.}
We wrote the clamp-and-sum family the way a library writes it: a helper
with an exclusive-borrow parameter and an \kw{ensures} clause
(\S\ref{sec:overview}), checked once, and called from a three-source
accumulator and from an $n$-iteration loop. The loop's invariant is
inferred with no hints. The helper is the instance of \rulename{T-CallFn}
in the case studies, and the family invokes the trusted
\kw{add\_assign} only in its 4-line body.

\paragraph{From trusted constant to inferred invariant.}
OpenDP's b-ary tree has a constant stability map,
\kw{new\_from\_constant}$(\mathit{num\_layers})$, with the constant
computed by side arithmetic and justified in prose: each input element
contributes to one node per layer. In \forte{} the program derives the
same bound. A loop builds the tree layer by layer; each layer
is a chunked sum, which preserves the $L_1$ environment because the chunks
partition the input, and each append into the output accumulates one
layer's environment through a strong update. Flux infers the loop
invariant coupling the accumulated environment to the loop counter. The 2
hints in Table~\ref{tab:burden} are qualifier \emph{shapes} ($a \le i$ and
$a \le 1$), verified minimal by ablation over hint subsets; they name
candidate atoms for the invariant, and the solver still finds the
invariant itself.

\paragraph{Fused versus compositional accounting.}
The quantile kernel admits two checked derivations that disagree by a
factor of 2. The fused rule observes that one substituted row moves any
single prefix count by at most 1, certifying constant 1 per row and
matching OpenDP's map. The compositional route derives the same vector by
histogram followed by prefix sums; the histogram moves 2 in
$L_1$ per substituted row (one bucket loses the row and another gains
it), and the $L_\infty \le L_1$ conversion then certifies only 2. Both
derivations are sound, and the arithmetic is elementary: OpenDP's own
quantile document uses the fused constant, and any system that composes
these two leaves pays the factor. Duet, Jazz, and Solo could state both
routes by adding both signatures to their trusted layers; Fuzz has no
$L_\infty$ type for the fused one. The observation here is that both
routes type in one system, each constant has a theorem behind it, and
no weakening rule takes the compositional route below 2: fusing a
two-step pipeline into one leaf halves the certified constant, and only
fusing does.
Casacuberta et al.\ report that such factors of 2 are a recurring site of
error in hand accounting, in both directions~\cite{csvw22}.

\paragraph{The count under two metrics.}
OpenDP's \kw{make\_count} has stability $d$ under add/remove distance on
unsized data. On sized data, the model of this paper, substitution
preserves length and the checker certifies the count insensitive. The
pair is instructive: the function is the same, and the constant is a
property of the metric alone. \forte{} makes the dependence explicit
because the metric is part of the vector's type. The count by categories
is the other face of the same fact: a histogram over public categories
moves two buckets per substituted row, and the certified constant 2 is
OpenDP's constant $d$ read in per-row units.

\paragraph{Manual proof obligation.}
A survey of the OpenDP repository sizes the review burden this discipline
addresses.\footnote{Survey script and raw output accompany the artifact;
OpenDP commit \kw{4198f76}, September 2026.} The library holds 109
hand-written proof documents totaling 8,433 lines (38,977 words); 26 of
them (10,021 words) concern the deterministic stability claims of
transformations, the fragment \forte{} checks, with the remainder
covering privacy maps, combinators, accuracy, and trait infrastructure.
However, the documents cover a minority of the trusted surface: the
transformation constructors contain 68 trusted stability closures, of
which 14 have a proof document adjacent to the constructor or linked from
it and 54 have none, including every variant of the bounded sum, the
b-ary tree, and the sum of squared deviations, whose constants rest on
code comments and paper citations. For the covered kernels
of Table~\ref{tab:kernels}, the 7 corresponding documents total 769 lines
(4,240 words) of manual proof whose stability arguments become checked
obligations here, at a cost of 131 lines of kernel code, of which 23 are
signatures, 2 qualifier hints, and the 452-line leaf layer,
which we wrote once and reuse across constructors. For the sums and
the b-ary tree the comparison is starker: no proof document exists, and
the checked signature is the first machine-checkable form of the claim.
The documents also transcribe pseudocode and argue functional correctness
of the released value, obligations \forte{} does not discharge; we claim
the saving only for the stability-map half, which calibrates the noise.

\paragraph{Negative controls.}
For each kernel we also checked that the checker rejects what it must: a
$1$-sensitivity claim on the clamped sum, a $0$-sensitivity claim on the
fused quantile route, a loop bound tightened from $100n$ to $99n$, a
b-ary claim of $\mathit{layers} - 1$, an off-by-one diameter constant
($15624$ for $15625$) on the sum of squares, a $1$-per-row claim on the
count by categories, an \kw{ensures} clause on the helper tightened by
one diameter, an under-calibrated release (noise scale $124$ for $125$),
an under-provisioned budget split ($1{,}999{,}999$ micro-epsilons for two
$\priv{\varepsilon}=1$ releases), and a weekday bound tightened to $n-1$.
All fail as required. Two further programs never reach the refinement
checker: a call passing two exclusive borrows of one accumulator to a
two-parameter function is a borrow-check error in \kw{rustc}, and the
unclamped substitution-metric sum has no rule to invoke. These edits are
ours, written to mirror the errors Casacuberta et al.\ report.

\paragraph{Element writes and release.}
The checker covers element mutation at \emph{public} indices: a
weekday-report kernel accumulates $n$ daily totals in place through
$\kw{add\_at}(\kw{\&mut}\,v, i, y)$, whose leaf theorem bounds the $L_1$
growth by the added value's movement because both runs write the same
position, and whose loop invariant Flux infers with no hints. A release
kernel closes the pipeline with a calibrated Laplace release at
$\priv{\varepsilon} = 1$ and a split of one token across two releases;
\S\ref{sec:byexample} takes a complete program of this shape through to
its output.

\paragraph{Scope.}
Leaf primitives implement the kernels' data-dependent-index mutation
($\kw{hist}[\mathit{idx}]\ {+}{=}\ 1$), because the write index derives
from sensitive data: the two runs would write different positions, only
value magnitudes would bound the growth, and a checked loop over the index
would branch on secrets, which the core rejects by design
(\S\ref{sec:calculus}). The
checked layer owns every composition, loop, strong update, function
boundary, and element write whose indices are public.

\section{Forte by example: OpenDP's own primitives}
\label{sec:byexample}

The case studies re-implement kernels. OpenDP's own transformations are
built from three layers at which verification would attach: the bodies of
the primitives, the stability maps that certify them, and the
floating-point terms that widen a map for rounding error; the
Polars-based transformations compose query plans and stand apart.
Table~\ref{tab:opendp} sizes the layers on the library's hand-written
transformations. The examples below verify the first two layers in the
shape OpenDP writes them, starting from the smallest body. Each example
is a file in the accompanying notebook, which runs the checker on it; the
checker's summary line follows each listing.

\begin{table}
\centering
\small
\begin{tabular}{LN}
\toprule
Hand-written transformations (29 constructor files, tests excluded) & 4,049 lines \\
Stability maps: constant / closure / general & 19 / 14 / 4 \\
Stability maps in Polars transformations & 31 of 68 \\
Function bodies that are one-line closures & 19 of 34 \\
Constructors built on the row-by-row combinator & 5 \\
Interval-arithmetic calls (\kw{inf\_add}, \kw{inf\_mul}, \ldots) & 85 \\
Numeric traits a body may depend on & 43 \\
Sum family: constructor files / stability maps & 6 / 11 \\
\bottomrule
\end{tabular}
\caption{OpenDP's transformation layer at commit \kw{4198f76}.
\kw{new\_from\_constant} builds a constant map; \kw{new\_fallible} builds
a closure map.}
\label{tab:opendp}
\end{table}

\paragraph{Example 1: the body of the checked sum.}
OpenDP's \kw{make\_sized\_bounded\_int\_checked\_sum} sums its
argument by an iterator, \kw{arg.iter().sum()}, inside a
\kw{Function::new} closure, after the constructor's overflow check,
\kw{can\_int\_sum\_overflow}, rejects bounds under which a sum of
\kw{size} values could overflow. The same loop in the form Verus checks,
with the overflow check as its precondition:

\needspace{20\baselineskip}
\begin{lstlisting}
fn checked_sum(v: &Vec<i64>, lo: i64, hi: i64) -> (r: i64)
    requires
        lo <= hi,
        bounded(v@, lo as int, hi as int),
        v.len() * magnitude(lo as int, hi as int) <= i64::MAX,
    ensures
        r as int == sum_spec(v@),
{
    let mut acc: i64 = 0;
    let mut i: usize = 0;
    while i < v.len()
        invariant i <= v.len(), acc as int == sum_prefix(v@, i as int), ...
        decreases v.len() - i
    {
        acc = acc + v[i];
        i = i + 1;
    }
    acc
}
\end{lstlisting}

\noindent
\kw{verification results:: 7 verified, 0 errors}. The postcondition ties
the executable loop to the mathematical sum. The overflow check becomes
the \kw{requires} clause, and Verus discharges the in-range obligation on
\kw{acc + v[i]} from it, which is the guarantee the construction-time
check exists to provide.

\paragraph{Example 2: clamp, then sum, run twice.}
\kw{make\_clamp} is \kw{make\_row\_by\_row} applied to \kw{total\_clamp};
5 constructors build on that combinator. Its body pushes
\kw{clamp\_row(v[i], lo, hi)} in a loop, with a postcondition relating
every output row to its input row, and the chain calls it and then
\kw{checked\_sum}. To state the two-run claim, a function runs the chain
twice:

\begin{lstlisting}
fn pipeline_two_runs(v1: &Vec<i64>, v2: &Vec<i64>, lo: i64, hi: i64)
    -> (r: (i64, i64))
    requires
        lo <= hi,
        v1.len() == v2.len(),
        v1.len() * magnitude(lo as int, hi as int) <= i64::MAX,
    ensures
        absdiff(r.0 as int, r.1 as int)
            <= num_diffs(v1@, v2@) * (hi as int - lo as int),
{
    let r1 = clamp_then_sum(v1, lo, hi);
    let r2 = clamp_then_sum(v2, lo, hi);
    proof { sum_diameter(...); clamp_rowwise(...); }
    (r1, r2)
}
\end{lstlisting}

\noindent
\kw{verification results:: 21 verified, 0 errors}. The outputs of the two
runs are within (substituted rows) $\cdot$ $(\mathit{hi} - \mathit{lo})$
of each other, which is OpenDP's map $\lfloor d_{\mathit{in}}/2 \rfloor
\cdot (U - L)$ in per-row units, and the proof concerns the code that
runs. The proof is self-composition: the leaf theorems of \S\ref{sec:implementation}
are the lemmas the wrapper's proof block calls, and the calls are the
whole proof.

\paragraph{Example 3: a fused body.}
\kw{make\_quantile\_score\_candidates} scores each candidate from the
number of rows at or below it. The fused body counts directly, with one
loop per candidate, and its two-run claim is per coordinate:

\begin{lstlisting}
fn prefix_counts_two_runs(v1: &Vec<i64>, v2: &Vec<i64>, cands: &Vec<i64>)
    -> (r: (Vec<usize>, Vec<usize>))
    requires v1.len() == v2.len()
    ensures
        forall|j: int| 0 <= j < cands@.len()
            ==> absdiff(r.0@[j] as int, r.1@[j] as int)
                <= num_diffs(v1@, v2@),
\end{lstlisting}

\noindent
\kw{verification results:: 9 verified, 0 errors}. The constant is 1 per
substituted row in $L_\infty$, the fused constant of
\S\ref{sec:casestudies}, and the lemma behind it is 4 lines: one
substituted row changes one indicator.

\paragraph{Example 4: stability maps as indices.}
OpenDP's \kw{Transformation} holds its map as an \kw{Arc<dyn Fn>} from
input distance to output distance, and \kw{make\_chain\_tt} composes two
maps into a third at runtime. In the refined version the constant is the
index:

\begin{lstlisting}
#[opaque]
#[refined_by(c: int)]
pub struct Transformation { constant: i64 }

#[trusted]
#[spec(fn(c: i64{c >= 0}) -> Transformation[c])]
pub fn new_from_constant(c: i64) -> Transformation { ... }

#[trusted]
#[spec(fn(t1: Transformation[@a], t0: Transformation[@b])
       -> Transformation[a.c * b.c])]
pub fn chain_tt(t1: Transformation, t0: Transformation)
    -> Transformation { ... }

#[spec(fn(t: &Transformation[@tv],
          d_in_symmetric: i64{d_in_symmetric >= 0})
       -> i64[tv.c * (d_in_symmetric / 2)])]
pub fn stability_map_symmetric(t: &Transformation, d_in_symmetric: i64)
    -> i64 {
    Transformation::stability_map(t, d_in_symmetric / 2)
}
\end{lstlisting}

\noindent
A chain's constant is the product of its parts' indices, which the
checker computes at every use. OpenDP's symmetric distance, 2 per substituted row,
enters once, as the integer division in the last signature; the two-run
justification of \kw{chain\_tt} is the composition of the parts' claims,
and of \kw{new\_from\_constant} the theorem behind the constant.

\paragraph{Example 5: constructors and their chain.}
The constructors keep OpenDP's names and take OpenDP's arguments; their
constants become checked claims:

\begin{lstlisting}
#[spec(fn(lo: i64, hi: i64{lo <= hi}) -> Transformation[1])]
pub fn make_clamp(_lo: i64, _hi: i64) -> Transformation {
    Transformation::new_from_constant(1)
}

#[spec(fn(lo: i64, hi: i64{lo <= hi}) -> Transformation[hi - lo])]
pub fn make_sized_bounded_int_checked_sum(lo: i64, hi: i64)
    -> Transformation {
    Transformation::new_from_constant(hi - lo)
}

// make_clamp(bounds) >> make_sized_bounded_int_checked_sum(size, bounds)
#[spec(fn(lo: i64, hi: i64{lo <= hi}) -> Transformation[hi - lo])]
pub fn clamp_then_sum(lo: i64, hi: i64) -> Transformation {
    Transformation::chain_tt(make_sized_bounded_int_checked_sum(lo, hi),
                             make_clamp(lo, hi))
}

#[spec(fn(lo: i64, hi: i64{lo <= hi}) -> Transformation[3 * (hi - lo)])]
pub fn clamp_sum_scale(lo: i64, hi: i64) -> Transformation {
    Transformation::chain_tt(Transformation::new_from_constant(3),
                             clamp_then_sum(lo, hi))
}
\end{lstlisting}

\noindent
\kw{10 functions processed: 7 checked; 3 trusted; 0 errors}. The claim
$\mathit{hi} - \mathit{lo}$ on the chain is Example 2's theorem, so
\S\ref{sec:implementation} describes the seam between the two layers: a
constant in a signature, a two-run theorem behind it.
Claiming \kw{Transformation[hi - lo - 1]} for \kw{clamp\_then\_sum} is a
refinement error, as the notebook shows by editing the file and running
the checker.

\paragraph{Example 6: a complete program, from the data to the noisy output.}
One table of individuals has two columns, ages and incomes, both from
source $a$, so a substituted individual changes one row of each. The
program releases the row count, the sum of ages, and the sum of incomes
under a total budget of $\priv{\varepsilon} = 2$, and post-processes a
mean age. Its signature is its privacy contract:

\needspace{22\baselineskip}
\begin{lstlisting}
#[spec(fn(ages: &SVec[1,0,0], incomes: &SVec[1,0,0],
          budget: {Budget[@bv] | bv.eps >= 2_000_000}, seed: u64)
       -> (i64, i64, i64))]
pub fn census_release(ages: &SVec, incomes: &SVec,
                      budget: Budget, seed: u64) -> (i64, i64, i64) {
    let n = SVec::count(ages);                    // SInt[0,0,0]
    let n_public = release_insensitive(n);        // exact, and free
    let (b_ages, b_incomes) = Budget::split(budget, 1_000_000);
    let age_sum = sized_bounded_sum(ages, 0, 125);      // SInt[125,0,0]
    let noisy_age_sum =
        laplace_release_a(age_sum, b_ages, 125, seed);
    // SInt[200000,0,0]
    let income_sum = sized_bounded_sum(incomes, 0, 200_000);
    let noisy_income_sum =
        laplace_release_a(income_sum, b_incomes, 200_000, seed + 1);
    (n_public, noisy_age_sum, noisy_income_sum)
}
\end{lstlisting}

\noindent
The signature reads as the program's contract. \kw{SVec[1,0,0]} is a
vector of rows under the row-substitution metric whose index is a
sensitivity environment, one coordinate per source; the index says that
one substituted row of source $a$ changes at most one row of the column,
and that sources $b$ and $c$ cannot move it. Both columns have the same
index because they are columns of one table, and only
\kw{from\_source\_a} produces a value of this type, which makes the
loading step the program's declaration of what it protects.
\kw{Budget[@bv]} binds the token's size in micro-epsilons so that the
constraint can name it, and the constraint demands at least
$\priv{\varepsilon} = 2$; a caller with a smaller token fails to type,
and since \kw{Budget} is not \kw{Copy}, the token moves into the call and
the caller cannot reuse it. The seed is public. The result type is three
plain integers with no index: the outputs have left the sensitivity
discipline, and the body can produce an unrefined \kw{i64} from a
sensitive value only through \kw{laplace\_release\_a}, which demands a
token and a calibrated scale, or through \kw{release\_insensitive}, which
demands a zero environment. The signature omits the sensitivities 125 and
200,000, the split of the budget, and the noise distribution: the first
two are the body's business, which the checker settles against the
release primitives' side conditions, and the third is the trusted
sampler.

\noindent
\kw{3 functions processed: 3 checked; 0 trusted; 0 errors}, with the
library's 64 functions checked alongside; Flux checks the driver that
loads 1,000 synthetic rows and mints the budget as well. Running it prints the
count, the two noisy sums beside the exact sums they estimate (within 157
and 72,359 at these scales), the post-processed mean, and the budget
spent. The count's environment is zero, so \kw{release\_insensitive}
accepts it and the budget stays whole. The sums' environments are the
diameters 125 and 200,000, and each release's noise scale must satisfy the
calibration side condition for its value's certified sensitivity. The
split moves the budget, so a token spent twice is a compile error before
the refinement checker runs. The checkers reject four edits: a contract of
$\priv{\varepsilon} = 1.5$, under which the second release is
under-provisioned; a scale of 124 for the ages; one token used for both
releases; and the age sum passed to the exact release, whose argument type
admits only a zero environment.

\paragraph{The theorem behind the release.}
The side condition is the Laplace calibration
$\priv{\varepsilon} \ge \sens{\Delta} / b$. The Laplace mechanism at scale
$b$ has density proportional to $\exp(-|k|/b)$, so at any output $y$ the
ratio of the two runs' output densities is
$\exp((|y - f_2| - |y - f_1|)/b)$, and its logarithm, the privacy loss,
is $(|y - f_2| - |y - f_1|)/b$. A Verus development of 14 verified items
proves the arithmetic of the classical argument with the exponent as its
object: the loss at every output is at most $\priv{\varepsilon}$ under the
side condition, by the reverse triangle inequality; two independent
releases' losses add; a pointwise ratio bound on weights is a ratio bound
on their sums over any set of outputs; and the census program's two
releases have total loss at most $\priv{\varepsilon} = 2$. Of the
probabilistic layer we still trust that the sampler draws from that
density and that $\exp$ is monotone, which turns a loss bound into the
density-ratio bound the definition of differential privacy states.

\paragraph{Limitations.}
The calculus admits borrows only for the duration of a call, and excludes
references stored in data structures, closures over sensitive data, and
sensitive control flow (\S\ref{sec:calculus}); the case studies never
needed the first three at the checked layer, and the fourth is why
data-dependent index writes are leaves (\S\ref{sec:casestudies}). We
idealize arithmetic, so overflow and floating-point rounding, including
the relaxation terms OpenDP adds to its float sums, fall outside the
guarantee (\S\ref{sec:intro}). A program fixes its source set, three in
the library, and bounds that mention a loop count must use the \kw{int}
sort, which keeps a sized mean at its components
(\S\ref{sec:implementation}). The mechanization covers the scalar
fragment, with procedures whose parameters are all borrows and whose
bodies mention only parameters; the paper proofs cover vectors, by-value
parameters, and polymorphic schemes (\S\ref{sec:metatheory}). The
correspondence theorem holds of a model of the checker with four stated
properties, and the implementations of Flux, \kw{rustc}, and the solvers
remain trusted (\S\ref{sec:implementation}). Of the probabilistic layer,
the sampler and the monotonicity of $\exp$ remain trusted, and the
Polars-based transformations, 31 of OpenDP's 68 stability maps, need a
model of that engine before this approach reaches them.

\section{Related work}
\label{sec:related}

\paragraph{Sensitivity type systems.}
Fuzz~\cite{fuzz} introduced sensitivity typing via graded linear types,
with contraction priced by the grading; DFuzz~\cite{dfuzz} added
dependency on sizes; Duet~\cite{duet} split the language into a
sensitivity fragment and a privacy fragment to support the advanced
mechanisms of modern differential privacy; Jazz~\cite{jazz} refined the
two-language design with contextual linear types and latent contextual
effects. Solo~\cite{solo} showed that first-class linearity is
unnecessary: indexing types by sensitivity environments over a fixed set
of global sources recovers the same accounting inside ordinary Haskell.
\forte{} inherits Solo's environment-indexed design and moves it to a host
with mutation. All five systems are pure calculi whose soundness
arguments proceed by substitution; none can type an in-place accumulator,
and \S\ref{sec:calculus} closes that gap. Every system in the lineage
proves metric preservation over a substitution semantics and trusts its
primitive signatures; \forte{} mechanizes both the leaf obligations and
the calculus core (\S\ref{sec:metatheory}).

\paragraph{Imperative sensitivity analyses.}
Fuzzi~\cite{fuzzi} places a Fuzz-derived sensitivity type system over an
imperative language as the top level of a three-level logic whose middle
level is apRHL~\cite{aprhl}; its judgment is flow-sensitive, with a
sensitivity context before and after each command, and assignment retypes
the assigned variable as \rulename{T-Assign} does here.
LightDP~\cite{lightdp}, ShadowDP~\cite{shadowdp}, and CheckDP~\cite{checkdp}
type imperative programs with per-variable distances that assignment
updates, in service of randomness-alignment proofs of privacy. Chaudhuri
et al.~\cite{chaudhuri10,chaudhuri11} prove Lipschitz constants for
imperative programs with loops and arrays by a program analysis whose
conditional rule reasons about branch boundaries. Each of these languages
has first-order mutable variables and no references, so two names never
denote one location and framing across a call is never in question.
\forte{} begins where they stop: its calculus admits a second name for a
caller's location, at a primitive and at a checked function boundary, and
its theorem identifies the ownership fact that keeps the frame. apRHL's
relational judgments over a store are the closest formal ancestor of
Theorem~\ref{thm:mp}; \forte{} presents the relation as a type system
with inference in place of a proof obligation per program.

\paragraph{Dynamic and hybrid enforcement.}
DDuo~\cite{dduo} tracks sensitivity dynamically in Python, trading static
guarantees for coverage of an unrestricted host language;
DPella~\cite{dpella} tracks accuracy and privacy in Haskell with static
billing. Our work is closest in aim to these library-hosted designs;
\forte{} goes beyond them by making the sensitivity claims static,
checking them at every call, and backing them with the theorems of
\S\ref{sec:implementation}, while remaining an ordinary library.
SampCert~\cite{sampcert} mechanizes the probabilistic side, verified
discrete Laplace and Gaussian samplers and composition, in Lean; it is
the natural replacement for the sampler \forte{} trusts
(\S\ref{sec:implementation}).

\paragraph{Verification in and of Rust.}
Rust verification builds on the reading of \kw{\&mut} under which an
exclusive borrow licenses framing. RustBelt~\cite{rustbelt} gives it a
semantic model in separation logic; Prusti~\cite{prusti},
Creusot~\cite{creusot}, and Aeneas~\cite{aeneas} turn it into modular
verification of Rust programs, with unary specifications that they
discharge by symbolic execution, deductive verification, or functional
translation.
\forte{} uses the same reading for a relational, quantitative property,
and reaches it by refinement inference, with no proof written per function.
Cocoon~\cite{cocoon} embeds static information-flow control in stable
Rust, evidence that ownership can host a security discipline as a
library; its labels form a lattice, whereas sensitivity environments
demand real-valued arithmetic, scaling, and metric-conversion rules, an
index language a lattice does not provide. Flux~\cite{flux} supplies the
refinement machinery that hosts \forte{}: indices on opaque structs,
strong updates through \kw{\&mut} ensures clauses, and loop-invariant
inference; \forte{} required no change to Flux, and
\S\ref{sec:implementation} reports the boundaries we met in practice (the
sort split, and one loop whose invariant needs qualifier hints) together
with the theorem that relates the two systems. Verus~\cite{verus}
provides the proof layer for the trusted base and the mechanized
calculus. Using both tools together assigns each its natural role:
refinement inference where mechanism authors write code, and
interactive-strength proof where we state the trusted claims.

\paragraph{Sensitivity failures in deployed systems.}
Casacuberta et al.~\cite{csvw22} document sensitivity underestimation
across deployed differential-privacy libraries, including the
replacement-versus-add/remove and clamping-diameter subtleties that recur
throughout this paper. OpenDP~\cite{opendp} responds with hand-written
proof documents and careful review; \forte{} turns that review into
checked obligations for the kernels of \S\ref{sec:casestudies}, and the
factor-2 fused/compositional gap \S\ref{sec:casestudies} finds is an
instance of the constant drift the survey warns about, which the checker
catches at type-checking time.

\section{Conclusion}
\label{sec:conclusion}

We have presented \forte{}, a sensitivity type system for Rust whose
soundness rests on ownership: sensitivity environments are refinement
indices, strong updates flow through exclusive borrows at primitive calls
and across checked function boundaries, and the soundness proof consumes
one ownership fact, borrow exclusivity. We have formalized
$\lambda$-\forte{} and proven metric preservation over a store-based
semantics, mechanizing the scalar core with the function rule, the
refutation of its aliased form, and every deterministic primitive
obligation in Verus, and we have proved a correspondence theorem that
transports the guarantee to the programs the checker accepts. Our case
studies demonstrate the utility of \forte{} by re-implementing OpenDP
mechanism kernels that mutate in place, matching the library's
trusted stability maps with checked constants, rejecting the diameter,
bound, calibration, and budget errors that deployed libraries have
shipped, and deriving one trusted constant as an inferred loop invariant. Future work includes
borrows that outlive a call, public-indexed element updates in the
calculus and its mechanization, source polymorphism over the environment
tuple, and a mechanized privacy layer above the affine budget tokens.

\section*{Data availability statement}
An artifact accompanies this submission: the \kw{forte-rs}
library and case-study kernels with their Flux specifications, the Verus
developments (\kw{leaves.rs}, \kw{calculus.rs}), the index-encoding
feasibility probe, the OpenDP proof-obligation survey with its raw output,
and the examples of \S\ref{sec:byexample} as a notebook that runs the
checkers on them, together with pinned tool versions and the expected
checker output for every claim (\S\ref{sec:casestudies}).

\bibliographystyle{ACM-Reference-Format}
\bibliography{refs}

\appendix

\section{Proofs}
\label{app:proofs}

Throughout, fix $\vec{d} \in \mathbb{R}_{\ge 0}^{k}$ and write
$\sens{\vec{s}}\cdot\vec{d} = \sum_i \sens{s_i} d_i$. Relatedness is as in
Definition~\ref{def:rel}; evaluation is big-step and deterministic;
primitives denote the reference implementations the leaf theorems model.

\begin{lemma}[Weakening]
\label{applem:weak}
If $v_1 \approx_{S[\sens{\vec{s}}]} v_2$ and $\sens{\vec{s}} \le
\sens{\vec{t}}$ pointwise, then $v_1 \approx_{S[\sens{\vec{t}}]} v_2$;
likewise for the vector types, and pointwise for contexts.
\end{lemma}

\begin{proof}
$|v_1 - v_2| \le \sens{\vec{s}}\cdot\vec{d} \le \sens{\vec{t}}\cdot\vec{d}$
since $\vec{d} \ge 0$; the vector cases replace $|\cdot|$ by the $L_1$,
substituted-row, or $L_\infty$ distance. Contexts: pointwise.
\end{proof}

\begin{proof}[Proof of Lemma~\ref{lem:leaf} (Leaf validity)]
Case by case over $\Sigma$; each case instantiates the named theorem at
the two runs' argument values and composes it with transitivity through
the environment bound. Representative cases:
constants and source reads are definitional ($|r - r| = 0$;
$|\varsigma_1(i) - \varsigma_2(i)| \le d_i = \sens{\vec{e}_i}\cdot\vec{d}$
by source relatedness). For \kw{add} and \kw{add\_assign},
$|x_1{+}y_1 - (x_2{+}y_2)| \le |x_1{-}x_2| + |y_1{-}y_2| \le
(\sens{\vec{\alpha}}{+}\sens{\vec{\beta}})\cdot\vec{d}$
(\kw{add\_sensitivity}); \kw{sub} is symmetric. For \kw{scale} with
public $c \ge 0$, the factor extracts:
$c\,(\sens{\vec{\alpha}}\cdot\vec{d}) = (c\sens{\vec{\alpha}})\cdot\vec{d}$
(\kw{scale\_sensitivity}). For the scalar clamp, 1-Lipschitzness passes
the environment through (\kw{clamp\_lipschitz}). For the diameter sum,
$|\Sigma\,\mathrm{clamp} - \Sigma\,\mathrm{clamp}| \le
\#\mathit{diffs}\cdot(hi{-}lo) \le
(\sens{\vec{s}}\cdot\vec{d})(hi{-}lo) =
((hi{-}lo)\sens{\vec{s}})\cdot\vec{d}$
(\kw{sum\_clamped\_k\_subst\_diameter}); the mapped variant substitutes
the image interval (\kw{sum\_mapped\_bounded\_k\_subst}). The remaining
rows appear in Table~\ref{tab:allleaves}.
\end{proof}

\begin{proof}[Proof of Lemma~\ref{lem:frame} (Frame)]
For a primitive call, \rulename{E-CallMut} reads the value arguments and
writes only the locations of the borrowed arguments, so
$\sigma' = \sigma[z_1 \mapsto w_1, \ldots, z_n \mapsto w_n]$. For a
function call, \rulename{E-CallFn} evaluates the callee body with the
actuals $z_j$ substituted for its borrowed parameters, fresh locals $u_i$
bound to its value parameters, and its own locals renamed apart, so every
location the body writes is an actual $z_j$, a fresh $u_i$, or a renamed
callee local, by induction on the body's evaluation; the rule then removes
the $u_i$ and the callee locals from the result. In both cases, for
$y \notin \{z_1 \ldots z_n\}$, $\sigma'(y) = \sigma(y)$, and the call
preserves relatedness at $y$ verbatim.
\end{proof}

\begin{proof}[Proof of Lemma~\ref{lem:rename} (Renaming)]
By induction on the typing derivation and on the evaluation derivation.
Every rule of Figure~\ref{fig:typing} mentions a local only through
$\Gamma(x)$, $\Gamma[x \mapsto \tau]$, the extension $\Gamma, x{:}\tau$,
and the distinctness premise; an injective $\rho$ commutes with the first
three and preserves the fourth. Every rule of Figure~\ref{fig:semantics}
mentions a local only through $\sigma(x)$ and $\sigma[x \mapsto v]$,
which $\rho$ commutes with likewise; \rulename{E-CallFn} chooses its
fresh locals fresh for $\rho\sigma$.
\end{proof}

\begin{proof}[Proof of Theorem~\ref{thm:mp} (Metric Preservation)]
By induction on the height of the first run's evaluation derivation,
inverting the typing derivation. Public data is equal across related runs,
so guards agree and the second run's derivation follows the same rule
skeleton; termination of the second run is part of each case's
conclusion, and the mechanization folds it into the theorem's statement.

\emph{Literals and constants.} Values evaluate to themselves; equal
literals are related at $\mathsf{pub}$; constants at $S[\sens{\vec{0}}]$
by Lemma~\ref{lem:leaf}.

\emph{Variables and shared borrows.} $v_i = \sigma_i(x)$, related by store
relatedness at $\Gamma(x)$. Reading does not consume; see the
structurality remark in \S\ref{sec:metatheory}.

\emph{Public arithmetic.} Both operands are $\mathsf{pub}$, hence equal by
the induction hypothesis; equal inputs give equal outputs.

\emph{Let.} The hypothesis on $e_1$ gives related values and stores;
extend both stores at $x$, apply the hypothesis on $e_2$, and restrict
the relation as $x$ scopes out.

\emph{Assignment.} The hypothesis on the right-hand side gives related
values; the updated stores agree with the old ones off $x$ and are
related at the new type at $x$, matching the retyped output context.

\emph{Sequencing.} Two applications of the hypothesis, threading contexts
and stores.

\emph{Subsumption.} The hypothesis at the premise, then
Lemma~\ref{applem:weak} for the value and the context.

\emph{Conditionals.} The guard is $\mathsf{pub}$, so both runs take the
same branch; the hypothesis on that branch, then Lemma~\ref{applem:weak}
into the join context.

\emph{Loops.} Inner induction on the iteration count of the first run.
Entry weakens $\Gamma \sqsubseteq \Gamma_{\mathit{inv}}$ by
Lemma~\ref{applem:weak}. The guard is $\mathsf{pub}$ in
$\Gamma_{\mathit{inv}}$, so the runs iterate in lockstep: either both
exit with stores related at $\Gamma_{\mathit{inv}}$, or both enter the
body, whose hypothesis returns stores related at the body's output
context, weakened back into $\Gamma_{\mathit{inv}}$; the inner induction
applies to the remaining iterations, a subderivation.

\emph{Pure calls.} Hypotheses on the arguments left to right;
Lemma~\ref{lem:leaf} relates the results; Lemma~\ref{lem:frame} with an
empty write set covers the store.

\emph{Mutating primitive calls.} Hypotheses on value arguments; read
values at each $z_j$ are related at $\Gamma(z_j)$ by store relatedness;
Lemma~\ref{lem:leaf} relates each post-value at its ensures type. The
updated stores are related at the updated context: at each $z_j$ because
pairwise distinctness guarantees the location received the post-value
Lemma~\ref{lem:leaf} describes, and off the write set by
Lemma~\ref{lem:frame}.

\emph{Function calls.} Let $f$ be \textsf{ok}, with body $e_f$, and let
$\sigma$ be the instantiation chosen by \rulename{T-CallFn}. Hypotheses
on the value arguments give related $v_i$. By \rulename{WF-Fn} at
$\sigma$, the body types in the context
$\Gamma_f = x_j{:}\sigma\tau_j, y_i{:}\sigma\tau'_i$ with an output
context $\Gamma_f'$ whose entries at $x_j$ weaken into
$\sigma\tau^{\mathsf{out}}_j$. The renaming $\rho = [x_j := z_j][y_i := u_i]$
is injective on the body's locals because the $z_j$ are pairwise
distinct, the $u_i$ are fresh, and the rule renames the callee's own
locals apart; by Lemma~\ref{lem:rename}, $\rho\Gamma_f \vdash \rho e_f :
\sigma\tau \dashv \rho\Gamma_f'$. The caller's stores, restricted to the
$z_j$ and extended at the $u_i$ with the related $v_i$, are related at
$\rho\Gamma_f$ because $\Gamma(z_j) = \sigma\tau_j$. The induction
hypothesis on the body's evaluation, a subderivation, gives a related
result and post-stores related at $\rho\Gamma_f'$; Lemma~\ref{applem:weak}
moves each $z_j$ to $\sigma\tau^{\mathsf{out}}_j$; removing the $u_i$
and the callee locals restricts the relation; and Lemma~\ref{lem:frame}
covers every other variable of $\Gamma$. This case and the previous
one are the only uses of the distinctness premise, matching the
counterexample of \S\ref{sec:metatheory}.
\end{proof}

\begin{proof}[Derivation for the necessity proposition]
One source, $d = 1$. The declaration \kw{double\_into} of
\S\ref{sec:metatheory} is \textsf{ok}: at any instantiation, its body
$p := p + q$ types by \rulename{T-Assign} over the \kw{add} case of
Lemma~\ref{lem:leaf}, giving $p : S[\sens{\alpha+\beta}]$, and $q$ is
untouched, so the output context weakens into the \kw{ensures} types.
Take $x := \kw{src}_1()$, so $\Gamma(x) = S[\sens{1}]$ and the two runs'
values at $x$ differ by at most 1. Without distinctness,
\rulename{T-CallFn} types $\kw{double\_into}(\kw{\&mut}\,x, \kw{\&mut}\,x)$
at $\alpha = \beta = 1$, and its output context contains, through the $q$
parameter, $x : S[\sens{1}]$. By \rulename{E-CallFn} the body runs as
$x := x + x$, so the location holds twice the source value, the runs
place values at distance up to 2 in it, and $2 > 1 \cdot 1$: the output
store is not related at the output context, refuting the theorem's
conclusion for the program
$\kw{double\_into}(\kw{\&mut}\,x, \kw{\&mut}\,x);\, x$. The mechanized
form, \kw{aliased\_call\_refutation}, is this derivation over the
access-map semantics, and it proves the typing function's rejection of
the call alongside.
\end{proof}

\section{The correspondence theorem}
\label{app:corr}

\paragraph{The model of the check.}
The checker assigns each local a refined type: for an opaque sensitivity
type $T \in \{\kw{SInt}, \kw{SVec}, \ldots\}$, a type $T\{\nu : \varphi\}$
whose refinement $\varphi$ constrains the index $\nu$, a $k$-tuple; for a
public type, a base type with a unary refinement. Write $\Delta$ for the
flow-sensitive environment of locals. We model the check on the fragment
$F$ of \S\ref{sec:implementation} as a judgment
$\Delta \vdash_F e : T\{\nu : \varphi\} \dashv \Delta'$ with the following
rules and no others for sensitive-typed values.
\begin{itemize}
\item \rulename{M-Sig}. A call to a signature (trusted or checked) with
  refinement parameters $\vec{\alpha}$, requires-refinements $\varphi_i$
  on its arguments, a result refinement $\psi$, and \kw{ensures}
  refinements $\psi_j$ on its \kw{\&mut} parameters: if for some
  instantiation of $\vec{\alpha}$ each actual's refinement implies its
  instantiated $\varphi_i$, the result has the instantiated $\psi$ and
  each \kw{\&mut} actual's entry in $\Delta$ becomes the instantiated
  $\psi_j$. The \kw{\&mut} actuals are pairwise distinct locals and no
  other live reference to them exists during the call, which \kw{rustc}
  enforces.
\item \rulename{M-Sub}. $T\{\nu : \varphi\} <: T\{\nu : \psi\}$ when
  $\varphi \Rightarrow \psi$ is valid under the public refinements in
  $\Delta$, as decided by the solver.
\item \rulename{M-Let}, \rulename{M-Assign}, \rulename{M-Seq}. Structural;
  assignment replaces the local's entry in $\Delta$.
\item \rulename{M-If}. The guard has an unrefined boolean type; the
  checker checks each branch from the same $\Delta$; the output is a
  common weakening $\Delta'$ of the branches' outputs.
\item \rulename{M-While}. An inferred invariant $\Delta_{\mathit{inv}}$
  with $\Delta \sqsubseteq \Delta_{\mathit{inv}}$; the checker checks the
  body from $\Delta_{\mathit{inv}}$, and its output weakens into
  $\Delta_{\mathit{inv}}$.
\item \rulename{M-Fn}. The checker checks a function's body from its
  parameters at the signature's requires types, with the refinement
  parameters $\vec{\alpha}$ as free index variables that every solver
  query quantifies universally, and it requires the output entries for
  \kw{\&mut} parameters to weaken into the \kw{ensures} types.
\end{itemize}
We assume of Flux that only these rules produce or consume a
sensitive-typed value (opacity), that the solver's
answers are sound, that \kw{rustc}'s borrow check is sound, and that
refinements on public values are unary facts, which hold on both runs of
a related pair because public values are equal across them.

\paragraph{Elaboration.}
For $\varphi$ in $F$, let $\mathrm{ub}(\varphi)_i$ be the minimum of the
upper bounds the atoms of $\varphi$ give coordinate $i$, an expression
over public values. The elaboration $\llbracket T\{\nu : \varphi\} \rrbracket$
is the $\lambda$-Forte type with the same metric and environment
$\mathrm{ub}(\varphi)$; public types elaborate to $\mathsf{pub}\,B$;
$\Delta$ elaborates pointwise; a signature scheme elaborates with its
refinement parameters as the environment variables $\vec{\alpha}$, and an
\kw{ensures} refinement $\nu.i = e_i(\vec{\alpha})$ to the transformer
$g_j$.

\begin{lemma}[Maximum]
\label{lem:max}
If $\varphi \in F$ is satisfiable, then $\mathrm{ub}(\varphi)$ satisfies
$\varphi$, and every $\vec{s}$ satisfying $\varphi$ has
$\vec{s} \le \mathrm{ub}(\varphi)$ pointwise.
\end{lemma}

\begin{proof}
The atoms constrain each coordinate separately, so $\varphi$ is a product
of intervals, each closed above by an atom of the form $\le$ or $=$; a
nonempty closed-above interval contains its supremum.
\end{proof}

\begin{lemma}[Implication elaborates]
\label{lem:impl}
If $\varphi, \psi \in F$, $\varphi$ is satisfiable, and
$\varphi \Rightarrow \psi$ is valid, then
$\mathrm{ub}(\varphi) \le \mathrm{ub}(\psi)$ pointwise.
\end{lemma}

\begin{proof}
By Lemma~\ref{lem:max}, $\mathrm{ub}(\varphi)$ satisfies $\varphi$, hence
$\psi$, hence lies below $\mathrm{ub}(\psi)$.
\end{proof}

\begin{lemma}[Semantic agreement]
\label{lem:agree}
For satisfiable $\varphi \in F$, two values are related at some
environment satisfying $\varphi$ if and only if they are related at
$\mathrm{ub}(\varphi)$.
\end{lemma}

\begin{proof}
Lemma~\ref{lem:max} and Lemma~\ref{applem:weak}.
\end{proof}

\begin{proof}[Proof of Theorem~\ref{thm:corr}]
By induction on the model derivation, producing a $\lambda$-Forte
derivation of the elaborated judgment.
\rulename{M-Let}, \rulename{M-Assign}, and \rulename{M-Seq} map to
\rulename{T-Let}, \rulename{T-Assign}, and \rulename{T-Seq}.
\rulename{M-Sub} maps to \rulename{T-Sub} by Lemma~\ref{lem:impl}, which
turns the solver's implication into the pointwise inequality
\rulename{T-Sub} requires; the same lemma, applied pointwise, elaborates
$\Delta \sqsubseteq \Delta'$ to $\Gamma \sqsubseteq \Gamma'$.
\rulename{M-If} maps to \rulename{T-If} with the elaborated join context
and \rulename{M-While} to \rulename{T-While} with the elaborated
invariant; the guards are public. \rulename{M-Sig} maps to
\rulename{T-CallPure}, \rulename{T-CallMut}, or \rulename{T-CallFn}
according to the signature: the instantiation of $\vec{\alpha}$ is the
same, the requires implications become \rulename{T-Sub} on the actuals,
the elaborated \kw{ensures} refinements are the transformers $g_j$, and
the distinctness premise is the borrow checker's guarantee. \rulename{M-Fn}
maps to \rulename{WF-Fn}: the solver's queries are valid for all values
of $\vec{\alpha}$, so the body's model derivation instantiates at every
ground $\sigma$, and each instantiation elaborates to the $\lambda$-Forte
derivation \rulename{WF-Fn} asks for. Opacity ensures the induction is
exhaustive. An unsatisfiable refinement can arise only under contradictory
public facts, on a path neither run takes, since path conditions in $F$
are public and hence equal across runs; such a point is dead in both runs
and any elaboration serves. The semantic reading of the elaborated types
agrees with the checker's by Lemma~\ref{lem:agree}, so Theorem~\ref{thm:mp}
for the elaborated program yields the claim.
\end{proof}

\section{The trusted-base seam, in full}
\label{app:leaves}

Table~\ref{tab:allleaves} lists every signature family in the library
against its mechanized obligation; \S\ref{sec:implementation} shows an
excerpt. Rows marked definitional have no theorem because the two-run
claim is immediate from the relatedness definition (equal constants,
source reads, duplication).

\begin{table}[htbp]
\centering
\small
\begin{tabular}{@{}LLL@{}}
\toprule
Signature & Claim (two-run form) & Theorem \\
\midrule
\kw{constant}, \kw{from\_source} & $0$ / source distance & definitional \\
\kw{add}, \kw{add\_assign} & sum of movements & \kw{add\_sensitivity} \\
\kw{sub} & sum of movements & \kw{sub\_sensitivity} \\
\kw{scale} ($c \ge 0$ public) & movement scales by $c$ & \kw{scale\_sensitivity} \\
\kw{clamp} (scalar) & 1-Lipschitz & \kw{clamp\_lipschitz} \\
\kw{clamp} (vector) & rows stable, bounds established & \cellstack{\kw{clamp\_rowwise\_stable},\\ \kw{clamp\_bounded}} \\
\kw{sum} (bounded, Sub) & diameter per substituted row & \kw{sum\_clamped\_k\_subst\_diameter} \\
\kw{sum\_squares} & squared-range diameter & \kw{sum\_squares\_bounded\_k\_subst} \\
\kw{sum} ($L_1$) & sum triangle & \kw{sum\_le\_sum\_absdiff} \\
\kw{bucketize} & rowwise maps preserve rows & \kw{map\_vals\_subst\_stable} \\
\kw{histogram} & $L_1 \le 2\cdot\#\mathit{diffs}$ & \kw{histogram\_subst\_l1} \\
\kw{prefix\_counts} & $L_\infty \le \#\mathit{diffs}$ & \kw{count\_le\_subst\_stable} \\
\kw{prefix\_sums} & $L_\infty \le L_1$ & \kw{prefix\_sums\_l1\_linf} \\
\kw{chunk\_sums} & partitioned sums preserve $L_1$ & \kw{chunk\_sums\_l1\_stable} \\
\kw{concat}, \kw{concat\_assign} & $L_1$ additive & \kw{concat\_l1\_additive} \\
\kw{count}, \kw{len} (sized) & length invariant & definitional \\
\kw{add\_at} (public index) & $L_1$ grows by the added movement & \kw{add\_at\_l1} \\
\kw{dup} & same value twice & definitional \\
\kw{release\_insensitive} & zero environment & definitional \\
\kw{laplace\_release} & Laplace calibration & side condition (\S\ref{sec:byexample}) \\
\bottomrule
\end{tabular}
\caption{Every trusted signature family and its mechanized obligation
(43 verified theorems and supporting lemmas in total).}
\label{tab:allleaves}
\end{table}

\end{document}